\documentclass{article}[11pt] 

\usepackage{amsmath} 
\usepackage{amssymb}
\usepackage{tikz}
\usetikzlibrary{arrows,shapes,shapes.misc,shadows,shadings,decorations.markings,automata}

\tikzset{cross/.style={cross out, draw=blue, minimum size=2*(#1-\pgflinewidth), inner sep=0pt, outer sep=0pt},
cross/.default={1.5pt}}

\usepackage{verbatim}
\usepackage{changepage}
\usepackage{pgfplots}

\usepackage[hmarginratio=1:1,top=32mm]{geometry} 

\usepackage{enumitem} 

\usepackage{titling} 
\usepackage{amsmath}
\usepackage{amsthm}
\usepackage{hyperref} 

\usepackage{caption}
\usepackage[labelformat=simple]{subcaption}
\usepackage{graphicx}
\usepackage{xcolor}
\definecolor{red}{RGB}{146,0,0}
\definecolor{blue}{RGB}{0,109,219}
\definecolor{green}{RGB}{36,255,36}

\DeclareMathOperator{\Opt}{Opt}

\DeclareMathOperator{\Prob}{P}

\DeclareMathOperator{\SVC}{SVC}
\DeclareMathOperator{\MSVC}{MSVC}

\DeclareMathOperator{\Cov}{Cov}
\DeclareMathOperator{\Val}{Val}

\DeclareMathOperator*{\E}{E}

\pretitle{\begin{center}\Huge\bfseries} 
\posttitle{\end{center}} 
\title{Some Results on Approximability of Minimum Sum Vertex Cover} 
\author{%
{\textsc{Aleksa Stankovi\'c} \thanks{Research supported by the Approximability and Proof Complexity project funded by the Knut and Alice Wallenberg Foundation.}} \\[1ex] 
\normalsize Department of Mathematics \\ 
\normalsize KTH Royal Institute of Technology\\ 
\normalsize \href{mailto:aleksas@kth.se}{aleksas@kth.se} 
}
\date{\today} 
\renewcommand{\maketitlehookd}{%

\begin{abstract}
\noindent 
We study the Minimum Sum Vertex Cover problem, which asks for an ordering of vertices in a graph that minimizes the total cover time of edges. In particular, $n$ vertices of the graph are visited according to an ordering, and for each edge this induces the first time it is covered. The goal of the problem is to find the ordering which minimizes the sum of the cover times over all edges in the graph.
\par 
In this work we give the first explicit hardness of approximation result for Minimum Sum Vertex Cover. In particular, assuming the Unique Games Conjecture, we show that the Minimum Sum Vertex Cover problem cannot be approximated within $1.0748$. The best approximation ratio for Minimum Sum Vertex Cover as of now is $16/9$, due to a recent work of Bansal, Batra, Farhadi, and Tetali.

\par 
We also study Minimum Sum Vertex Cover problem on regular graphs. In particular, we show that in this case the problem is hard to approximate within $1.0157$. We also revisit an approximation algorithm for regular graphs outlined in the work of Feige, Lovász, and Tetali, to show that Minimum Sum Vertex Cover can be approximated within $1.225$ on regular graphs.
\end{abstract}
}

\hypersetup{
  colorlinks,
  citecolor=black,
  linkcolor=red,
  urlcolor=blue}

\theoremstyle{plain}
\newtheorem{thm}{Theorem}[section] 
\newtheorem{lemma}[thm]{Lemma} 
\newtheorem{theorem}[thm]{Theorem} 
\newtheorem{fact}[thm]{Fact} 
\newtheorem{conjecture}[thm]{Conjecture} 

\theoremstyle{definition}
\newtheorem{definition}[thm]{Definition} 

\begin{document}

\maketitle

\section{Introduction}
In the Minimum Sum Vertex Cover problem, as an input we are given a graph $G=(V,E)$, and the goal is to find an ordering of vertices which minimizes the total cover time of edges in $E$. In particular, we visit vertices in $|V|$ steps, one at each step, and an edge  $e$ is considered to be covered at the time $t\in \left \{1,\hdots,|V|\right\}$ if the first time one of its endpoints is visited by the ordering is $t$.
\par 
The Minimum Sum Vertex Cover (MSVC) problem was introduced by Feige, Lovász, and Tetali \cite{DBLP:journals/algorithmica/FeigeLT04}, as a special case of the Minimum Sum Set Cover problem, which was of primary interest in that work. The same work showed that MSVC can be approximated within a factor of $2$ using linear programming. That work also studied MSVC on regular graphs, and observed that a greedy algorithm approximates the optimal value within a factor of $4/3$. In addition to this, using semidefinite programming it was shown that the $4/3$ factor can be improved to some non-explicit constant $\beta$ smaller than $4/3$.
\par 
The $2$-approximation algorithm for MSVC was subsequently improved by Barenholz, Feige, and Peleg \cite{BFP06}, who gave a $1.999946$-approximation algorithm for this problem. This was then substantially improved by Bansal, Batra, Farhadi, and Tetali, who, using linear programming with fairly involved rounding procedure, showed that MSVC can be approximated within a factor of $16/9$. Furthermore, the same work gives a linear programming integrality gap matching the approximation ratio.
\par 
So far explicit hardness of approximation results for this problem have been lacking, and to the best knowledge of the author, the only inapproximability result~\cite{DBLP:journals/algorithmica/FeigeLT04} gives hardness of $1+\varepsilon$, for some small non-explicit $\varepsilon>0$, using a reduction from the Minimum Vertex Cover problem on bounded degree graphs \cite{DBLP:conf/focs/AroraLMSS92,DBLP:journals/toc/AustrinKS11}. In this work we give the first explicit hardness for MSVC, which we state in the following theorem.
\begin{theorem}\label{main_theorem}
	Assuming the Unique Games Conjecture, Minimum Sum Vertex Cover is NP-hard to approximate within $1.0748$.
\end{theorem}
We use the Unique Games Conjecture introduced by Khot \cite{DBLP:conf/stoc/Khot02a} as our hardness assumption. This conjecture has been a central open problem in the hardness of approximation area since its inception, and many already known (and optimal) hardness of approximation results rely on the validity of this conjecture \cite{DBLP:conf/stoc/Raghavendra08,DBLP:conf/stoc/Austrin07,DBLP:journals/jcss/KhotR08,DBLP:conf/focs/BansalK09}.  \par
We also study the MSVC on regular graphs, and in particular give the following inapproximability result
\begin{theorem}\label{regular_main_theorem}
	Assuming the Unique Games Conjecture, Minimum Sum Vertex Cover on regular graphs is NP-hard to approximate within $1.0157$.
\end{theorem}
Finally, we revisit the approximation algorithm of Feige, Lovász, and Tetali \cite{DBLP:journals/algorithmica/FeigeLT04} for regular graphs. Our contribution can be described as follows. The algorithm for regular graphs outlined in \cite{DBLP:journals/algorithmica/FeigeLT04} uses an approximation algorithm for a problem called Max-$k$-VC in a ``black box'' manner. Max-$k$-VC problem is the problem of finding $k$ vertices in a graph that cover as many edges as possible. The approximation ratio of the algorithm for regular graphs in \cite{DBLP:journals/algorithmica/FeigeLT04} depends on the approximation ratio $\alpha$ for Max-$k$-VC problem. Due to the developments since the publication of \cite{DBLP:journals/algorithmica/FeigeLT04} on Max-$k$-VC, a better value of $\alpha$ can be achieved, and hence by using this value we can obtain a stronger approximation. Furthermore, a certain bound\footnote{We do not discuss what this bound exactly is here, for the sake of clarity.} used in an argument outlined in \cite{DBLP:journals/algorithmica/FeigeLT04} for the approximation algorithm on regular graphs is incorrect, which we show by giving a counterexample in the appendix. We correct this by proving the optimal bound, and observe that the rest of the argument still holds. Let us remark that the sharpness of the bound affects the approximation ratio, and hence finding the optimal bound is desirable in this case. In conclusion, we obtain the following result
\begin{theorem}\label{reg_theorem}
	Minimum Sum Vertex Cover can be approximated within $1.225$ on regular graphs.
\end{theorem}

\subsection{Techniques and Proof Ideas}
In this section we give an overview of the proof and briefly discuss techniques used.   \par
The starting point of our reduction are Unique Games, which we formally describe in Section \ref{section:preliminaries}. More precisely, we use regular Affine Unique Games as an input to our reduction. Regular Affine Unique Games are Unique Games in which the alphabet is understood as an additive group $\mathbb{Z}_L$, we consider constraints of form $x_u-x_v = c_e$ for an edge $e=(u,v)$, while the word regular indicates that the constraint graph is regular. Interestingly, in this work the structure of Affine Unique Games actually helps us achieve better completeness and therefore a stronger inapproximability result. The property of Affine Unique Games that we use can be described as follows. Let us consider the completeness case, in which we have some assignment $z$ of labels to the vertices in the Affine Unique Games, which satisfies almost all the constraints. Then, for any $a\in \mathbb{Z}_L$, the assignment $z_a = z + a$ gives another assignment which satisfies almost all the constraints. Furthermore, if we let $V_a, a \in \mathbb{Z}_L,$ to be the subset of the label extended graph  comprised of vertex labels ``selected'' by the map $z_a$, then the sets $V_a=\left \{(v,z_a(v)\right\}$ are disjoint, and this gives us enough structure to find an ordering with a low sum set cover value.
\par
Let us elaborate. Our reduction uses the same standard long code dictatorship testing as the celebrated paper of Khot, Kindler, Mossel, and O'Donnell \cite{DBLP:journals/siamcomp/KhotKMO07}, which among other results gave the optimal hardness of Max-Cut assuming the Unique Games Conjecture. This is the same reduction that appeared in \cite{DBLP:journals/toc/AustrinKS11,DBLP:conf/approx/AustrinS19}, and hence the graphs that are output by the reduction satisfy the same properties as outlined in these works, which turns out to be useful for studying soundness. In particular, in the soundness case, for each $r \in (0,1)$, and each vertex subset of fractional size $r$, we have a lower bound $b:=b(r)$ on the number of edges with both endpoints in this subset. Therefore, no matter which order of visiting the vertices we choose, after $t\in \left \{1,\hdots,n\right\}$ steps, we have not covered edges which have both endpoints in vertices visited after the time $t$, and hence at the time $t$ we have at least $b(1-t/n)$ uncovered edges. This gives us a lower bound of form $\sum_{i=1}^{n} b(1-r/n) \approx \int_{0}^{1} b(x) dx$. 
\par
In the completeness case, we are supposed to specify an ordering of the vertices in each of $k \in \mathbb{N}$ long codes. Once such an ordering is fixed, in the first pass we would pick first vertex in each of the $k$ long codes, after which we would pick the second vertex in each long code, etc. The order in which we visit  $k$ long codes will not be impactful. Hence, it is very important to pick the order of visiting vertices in each long code well. This is where the affine structure of Unique Games is useful. In the case we have only one good labeling (as it is the case with ``classical'' Unique Games), an obvious observation is that we can take first all vertices with $0$ in the coordinate fixed by a good labelling $z=z_0$, and then all vertices with $1$ in the same coordinate. However, there are many vertices in a long code which have $0$ in the coordinate fixed by a good labelling, and hence many orderings can be chosen. Therefore, the question is which order should one pick the vertices with in this subset? Since in Affine Unique Games we have a second satisfying assignment, namely $z_1$, there is a natural ordering among these. We iterate through vertices that have $0$ in the coordinate fixed by $z_1$, and after visiting the whole subgraph, visit vertices that have $1$ in the coordinate fixed by $z_1$. We can repeat this idea and visit smaller and smaller subgraphs, the last of which will consist only of two vertices and for which we will use $z_{L-1}$.
\par
The idea of using multiple good assignments in reductions from Unique Games already appeared in \cite{DBLP:journals/corr/abs-2010-01459}, but it is still fairly uncommon. Hence, it would be interesting to see whether it would be useful for some other problems as well.
	\par 
The output of the hardness reduction is a weighted graph, and we need to remove its weights. The idea for this is simple: we replace each vertex $v$ with $m$ new vertices which we group in a set $A_v$, for $m$ is sufficiently large. We then replace each edge $e=(u,v)$ by sampling edges between $A_u$ and $A_v$ at a correct density. This graph indeed looks like the initial graph and is almost regular, however, proving that it preserves soundness and completeness properties, requires some effort. With some additional care we can also make this graph regular, and this yields the proof of Theorem \ref{regular_main_theorem}.
\par 
To obtain the result stated in Theorem \ref{main_theorem}, we use the idea that already appeared in \cite{DBLP:journals/algorithmica/FeigeLT04}, and it is originally inspired by \cite{DBLP:journals/ipl/Bar-NoyHK99}. We first fix a number $k \in \mathbb{N}$, $k$  weights $\alpha_1,\hdots,\alpha_k \in \mathbb{R}_+,$ and $k$  correlations parameters $\rho_1,\hdots,\rho_k \in (-1,0]$. We then create $k$ regular graphs $G_1,\hdots,G_k,$ similar to the ones used in Theorem \ref{regular_main_theorem}, where $i$-th graph $G_i=(V_i,E_i)$ depends on a correlation\footnote{How the graphs depend on correlation parameters will be explained in Section \ref{section:hardness_reduction}.} parameter $\rho_i$.  The vertex sets $V_i$ are disjoint, as are the edges. For each $G_i$, we scale the weights\footnote{Obviously, we can take without loss of generality that $\alpha_1=1$} of edges $E_i$ in $G_i$  by $\alpha_i$. Assuming the weights are decreasing, the optimal ordering in both the soundness and completeness case first picks some vertices in $G_1$ that cover most edges since this is the graph with the highest weight, after which the ordering will start picking vertices among both $G_1$  and $G_2$, and after a while vertices from many graphs will ``compete'' for their place in the ordering. Although we do not find analytically what the best ordering should be, we can numerically calculate this value, and we also use numerics\footnote{Namely grid search coupled with gradient descent} to find weights that give improved hardness, which results in the proof of Theorem \ref{main_theorem}. Actually the graph $G=G_1\cup \hdots \cup G_k$ will be weighted graph, but its weights can be ``removed'', which will be discussed in Section \ref{section:remove_weights}.
		
\subsection{Organization} 
In Section \ref{section:preliminaries} we introduce the notation used in this work, recall some well known facts, and formally introduce the Minimum Sum Vertex Cover problem. Then, in Section~\ref{section:hardness_reduction}, we give our hardness reduction to weighted graphs. We split this section into two parts, with Section \ref{section:UG_reduction} discussing the reduction from Unique Games to weighted graphs which are regular in ``weighted sense''. In Section \ref{subsection:non_regular_hardness} we discuss how we can use these instances to get increased inapproximability at the loss of regularity.
\par 
In Section \ref{section:remove_weights} we show how the MSVC on weighted graphs can be reduced to MSVC on unweighted graphs with up to $o(1)$  additive loss in approximation ratio, and also discuss how we obtain regular graphs from the reduction discussed in Section \ref{section:UG_reduction}. 
\par Finally, in Section \ref{section:regular_graphs} we show how Minimum Sum Vertex Cover on regular graphs can be approximated within a factor of $1.225$, by recalling the algorithm from \cite{DBLP:journals/algorithmica/FeigeLT04} and making necessary changes.

\section{Preliminaries}\label{section:preliminaries}
For $n \in \mathbb{N}$ we use $[n]$  to denote $[n] = \left \{1,2,\hdots,n\right\}$. In this paper we work with undirected (multi)graphs $G = (V, E)$. For a set $S \subseteq V$ of vertices we use $S^c$ to denote its complement $S^c = V \setminus S$, and write $U \sqcup V$ for a disjoint union of sets $U$ and $V$. 
\par 
The initial graph output by our reduction will be edge weighted. The weights of edges are given by a function $W\colon E \rightarrow \mathbb{R}_+$. For a subset  $K \subseteq E$ we interpret $W(K)$ as:
\begin{equation*}
	\sum_{e \in K}W(e),
\end{equation*}
and $w(K)=W(K)/W(E)$. We usually write $W_e$ for $W(e)$. For $S,T \subseteq V$, we write $W(S,T)$ for the total weight of edges from $E$ which have one endpoint in $S$, and other in $T$, and $w(S,T):= W(S,T)/W(E)$. Note that, since we work with undirected graphs, the order of endpoints is not important, and therefore $w(S,T) = w(T,S)$. We remark that the sets $S$ and $T$ do not need to be disjoint. We also use $N(S,T)$ to denote the set of all edges with one endpoint in $S$ and the other endpoint in $T$. For a vertex $v \in V$, we use $N(v)$ to denote the set of its neighbours. Sometimes we will be interested in weights over two different graphs that are defined on the same vertex set, in which case we will add the graph label to the subscript in the quantity we want to express. For example, if we have two graphs $G,G'$ over the same vertex set  $V$, we will write $w_G(S,T), w_{G'}(S,T)$, etc.
 \par 
 The following definition will be useful for discussing properties of our reduction.
\begin{definition}
  A graph $G$ is $(r,h)$-dense if every subset $S \subseteq V$ with $w(S) = r$ satisfies $w(S,S) \geq h$.
\end{definition}
 Minimum Sum Vertex Cover is arguably more natural in an unweighted setting, i.e., setting in which the weights of all edges are equal, and we introduce it in this setting first, before generalizing it to the weighted case.
\begin{definition}
	Consider an unweighted graph $G=(V,E)$, and let $n=|V|$. For an ordering of vertices represented as a bijection $\sigma\colon [n] \leftrightarrow V$, and an edge $(u,v)=e \in E$, let us denote with $c_{\sigma,e}$ the ``time'' at which edge $e$  is covered, that is
	\begin{equation*}
		c_{\sigma,e} = \min(\sigma^{-1}(u), \sigma^{-1}(v)).
	\end{equation*}
	Then the Sum Vertex Cover under scheduling $\sigma$, which we denote by $\SVC_G(\sigma)$,	is given as 
	\begin{equation}\label{SVC_def_1}
		\SVC_G(\sigma) = \sum_{e \in E} c_{\sigma,e}.
	\end{equation}
	The value of Min Sum Vertex Cover is the minimal value of $SVC_G(\sigma)$ over all possible permutations $\sigma$, that is 
	\begin{equation}\label{MSVC_def_2}
		\MSVC(G) = \min_{\sigma\colon  [n]  \leftrightarrow V} \SVC_G(\sigma).
	\end{equation}
\end{definition}
 We can also reformulate the expression \eqref{SVC_def_1}, stating the value of Sum Vertex Cover under scheduling $\sigma$, as follows. At the time $t \in [n]$, the total number of edges not covered\footnote{We interpret $\sigma([t])$ as $\sigma([t]) = \left \{ \sigma(i) \mid i \in [t] \right\}$.} is $W(\sigma([t])^c,\sigma([t])^c)$, and let us assign them the cost of $1$  at that time. 
The cost $c_{\sigma,e}$  of an edge $e$  under $\sigma$ is exactly the number of times $t$  the edge was not covered, and hence we can write 
\begin{equation}
	\SVC_G(\sigma) = \sum_{t=1}^n W(\sigma([t])^c,\sigma([t])^c).
\end{equation}
We remark that this naturally allows us to define Minimum Sum Vertex Cover for edge weighted graphs. We can also discuss Minimum Sum Vertex Cover for weighted graphs in the sense of definitions \eqref{SVC_def_1} and \eqref{MSVC_def_2} by letting 
\begin{equation*}
	\SVC_G(\sigma) = \sum_{e \in E} W_e c_{\sigma,e}.
\end{equation*}
We can extend this definition in a natural way to include vertex weights. However, we have not found vertex weights to be useful for hardness reduction, and hence we omit further discussing this for the sake of simplicity.
\par
In order to state the quantities appearing in our result, it is necessary to introduce some more notation. We use $\phi(x) = \frac{1}{\sqrt{2\pi}}e^{-x^2/2}  $ to denote the density function of a standard normal random variable, and $\Phi(x) = \int_{-\infty} ^{x} \phi(y)dy$ to denote its cumulative distribution function (CDF). We also work with bivariate normal random variables, and to that end introduce the following function.
\begin{definition}
  Let $\rho \in [-1,1] $, and consider two jointly normal random variables $X,Y,$ with mean $0$, and covariance matrix $\Cov(X,Y) = \begin{bmatrix}1  & \rho \\ \rho  & 1   \end{bmatrix}$. We define $\Gamma_{\rho} \colon [0,1]^2 \rightarrow [0,1] $ as 
  \begin{equation*}
    \Gamma_{\rho}(x,y) = \Pr \left [ X \leq \Phi^{-1}(x) \wedge Y \leq \Phi^{-1}(y) \right ]. 
  \end{equation*}
We also write $\Gamma_{\rho}(x) = \Gamma_{\rho}(x,x)$. 
\end{definition}


Knowing the first derivative of $\Gamma_{\rho}$ will be useful in our proofs, and for that reason we state it here.
\begin{fact}\label{fact:rho_der}
	For $\rho \in (-1,1)$  we have 
	\begin{equation*}
		\frac{\partial  }{\partial r} \Gamma_{\rho}(r) = - \Phi\left(\sqrt{ \frac{1-\rho}{1+\rho}} \Phi^{-1}\left(\frac{1-r}{2}\right) \right).
	\end{equation*}
	Hence, $\left|\frac{\partial  }{\partial r} \Gamma_{\rho}(r)\right| \leq 1$,  for all $r\in [0,1], \rho \in (-1,1)$.
\end{fact}


The proof of this fact can be found in \cite{AustrinThesis}, as Proposition 3.6.3. The hardness result stated in this paper is based on the Unique Games Conjecture. In order to state this conjecture, we first introduce Unique Games.
\begin{definition}
  A \emph{Unique Games instance} $\Lambda = (\mathcal{U},\mathcal{V},\mathcal{E},\Pi,[L])$ consists of an unweighted bipartite multigraph 
  $(\mathcal{U} \sqcup \mathcal{V},\mathcal{E})$, a set 
  $\Pi = \{\pi_{e}\colon [L] \to [L] \mid e \in \mathcal{E} \textrm{ and }\pi_e \textrm{ is a bijection} \} $ of permutation constraints, and a set $[L]$ of labels. The value of $\Lambda$ under the assignment $z\colon \mathcal{U} \sqcup \mathcal{V} \to [L] $ is the fraction of edges satisfied, where an edge $e=(u,v), u \in \mathcal{U}, v \in \mathcal{V},$ is satisfied if $\pi_e(z(u)) = z(v)$. We write $\Val_z(\Lambda)$ for the value of $\Lambda$ under $z$, and $\Opt(\Lambda)$ for the maximum possible value over all assignments $z$.
\end{definition}
Let us remark that we require Unique Games instance graph $\left(\mathcal{U},\mathcal{V},\mathcal{E}\right)$ to be regular. Since Unique Games belong to the class of problems known as Constraint Satisfaction Problems (CSPs), without loss of generality we can assume regularity, as shown in \cite{DBLP:journals/ipl/Stankovic22}. \par
The Unique Games Conjecture  \cite{DBLP:conf/stoc/Khot02a} can be stated as follows (\cite{DBLP:conf/coco/KhotR03}, Lemma 3.4).
\begin{conjecture}[Unique Games Conjecture] 
  For every constant $\gamma > 0$ there is a sufficiently large $L \in \mathbb{N}$, such that for a Unique Games instance $\Lambda = (\mathcal{U},\mathcal{V},\mathcal{E},\Pi,[L]) $ with a regular bipartite graph $(\mathcal{U} \sqcup \mathcal{V}, \mathcal{E})$, it is NP-hard to distinguish between
  \begin{itemize}
    \item $\Opt(\Lambda) \geq 1-\gamma$, 
    \item $\Opt(\Lambda) \leq \gamma$.
  \end{itemize}
\end{conjecture}
The starting point of hardness result in this work are variant of Affine Unique Games, which are defined as follows.
\begin{definition}
	An Affine \emph{Unique Games instance} $\Lambda = (\mathcal{U},\mathcal{V},\mathcal{E},\Pi,[L])$ is a Unique Games Instance $\Lambda$ in which all permutation constraints $\pi_e$ are affine constraints. Furthermore, the alphabet $[L]$ is identified with an additive group $\mathbb{Z}_L$, and for each $\mathcal{E} \ni e = (u,v) $  we have $\pi_e(x)=x-c_e$, where $c_e \in \mathbb{Z}_L$ is a constant.
\end{definition}
We remark that approximating Affine Unique Games is equally hard as approximating Unique Games, in the sense stated by the lemma below which was proved in \cite{DBLP:journals/siamcomp/KhotKMO07}.
\begin{lemma}[Affine Unique Games Hardness] 
	Assuming the Unique Games Conjecture, the following statement holds. For every constant $\gamma > 0$, there is a sufficiently large $L \in \mathbb{N}$, such that for an Affine Unique Games instance $\Lambda = (\mathcal{U},\mathcal{V},\mathcal{E},\Pi,[L]) $ with a regular bipartite graph $(\mathcal{U} \sqcup \mathcal{V}, \mathcal{E})$, it is NP-hard to distinguish between
  \begin{itemize}
    \item $\Opt(\Lambda) \geq 1-\gamma$, 
    \item $\Opt(\Lambda) \leq \gamma$.
  \end{itemize}
\end{lemma}
We will also use the following version of Hoeffding's inequality \cite{doi:10.1080/01621459.1963.10500830}.
  \begin{lemma}\label{hoeff_bound}
    Let $X_1,\hdots,X_k$ be independent random variables such that the range of each $X_i$ is $[0,b]$, where $b \in \mathbb{R}$. Then for $X=\sum_{i=1}^k X_i$ we have 
		\begin{equation}\label{hoef_bound_eq}
      \Prob[ |X-\E[X] | \geq t ] \leq 2 e^{-\frac{t^2}{k b^2}}.
    \end{equation}
  \end{lemma}

\section{Hardness Reduction}\label{section:hardness_reduction}
In this section we give a hardness reduction from Unique Games to the weighted graphs.
In Section \ref{section:UG_reduction} we give a reduction from Affine Unique Games to weighted graphs which will contain properties that will be used for showing hardness of approximating Min Sum Vertex Cover on regular graphs (Theorem \ref{regular_main_theorem}). However, this graphs will be weighted, and the final step will be given in Section \ref{section:remove_weights}, by showing that a family of unweighted regular graphs exists with the similar properties.
\par 
We also discuss how the inapproximablity given by this family of graphs can be amplified at the expense of regularity in Section \ref{subsection:non_regular_hardness}. Once again we will show in Section \ref{section:remove_weights} how the weights on these graphs can be ``removed'', i.e., a family of unweighted graphs with \emph{essentialy} the same properties exists, and this will conclude the proof of Theorem \ref{main_theorem}.
\subsection{Reduction from Unique Games to Regular Weighted Graphs} \label{section:UG_reduction}
We remark that we use the same type of reduction as in \cite{DBLP:journals/siamcomp/KhotKMO07,DBLP:journals/toc/AustrinKS11,DBLP:conf/approx/AustrinS19}, with the only difference being that we now use Affine Unique Games as the starting point, and compared to \cite{DBLP:conf/approx/AustrinS19} we are here interested only\footnote{We remark that one could also consider using a reduction with biased bits, i.e., the reduction from \cite{DBLP:conf/approx/AustrinS19} with $q \neq 1/2$. However, this does not yield better inapproximability.} in the unbiased setting ($q=1/2$). The main challenge lies in proving completeness, since we will reuse the soundness property of the reduction in the aforementioned results. 
\par Before giving the full proof of the result, we will sketch the ideas behind studying the completeness case now. Consider having a labelling $z$ which satisfies almost all the edges. Let us describe what happens locally on two vertices $u,v,$ with a common neighbour $w$, which are chosen such that $(u,w)$ and  $(v,w)$ edges are satisfied by $z$. For the sake of simplicity, let us assume that the affine constraints on $e_1=(u,w)$, and  $e_2=(v,w)$, are trivial, that is,  $c_{e_1}=c_{e_2}=0$, so that the labels  $z(u)$ and  $z(v)$ are matched if and only if $z(u)=z(v)$. Then, we replace both $u$ and $v$ with $2^{L}$ strings of length $L$. Let us call the sets of strings which replaced $u$ and $v$ as $R$ and $S$, respectively. We drop indices $u,v,$ here, for the sake of readability. Hence, we have  
	\begin{equation*}
		S= \left \{(s_1,\hdots,s_L) \mid s_i \in \left \{0,1\right\}, i \in [L]\right\}, \quad R= \left \{(r_1,\hdots,r_L) \mid r_i \in \left \{0,1\right\}, i \in [L]\right\}. 
	\end{equation*}
	Edges between $S$ and $R$ are created as follows. The reduction first fixes some negative correlation parameter $\rho \in (-1,0)$. It then samples $L$ pairs of unbiased, $\rho$ correlated bits, $(s_i,r_i), i=1,\hdots,L$, and then adds an edge between $s=(s_1,\hdots,s_L) \in S$ and $r=(r_1,\hdots,r_L) \in R$. Let us use $\nu$ to denote the probability distribution of  two $\rho$ correlated, unbiased bits, i.e., 
    \begin{equation*} 
			\nu(0,0)=\nu(1,1) = \frac{1+\rho}{4}, \quad \nu(0,1) = \nu(1,0) = \frac{1-\rho}{4},
    \end{equation*}
		and study Minimum Sum Vertex Cover on this graph. We will upper bound the value of MSVC on this graph $G_L$ by\footnote{Without loss of generality, we assume that weights of edges sum up to $1$ here.} some $T_L$, by exhibiting an ordering $\sigma_L$. Actually, we build our ordering for vertices in $G_L$ by using the ordering on $G_{L-1}$, which is a graph that would have been created with an alphabet size $L-1$. In particular, we observe that the induced subgraph of $G_L$ obtained by fixing $s_1=r_1=0$ is isomorphic to  $G_{L-1}$. Hence, if we use $\sigma_{L-1}$ to visit vertices in this subgraph, edges with both endpoints in it will be visited by the time $T_{L-1}$ on average. Since the total weight of edges in this subgraph is $\nu(0,0)$, the cost of covering edges in this subgraph is at most 
		\begin{equation*}
			\nu(0,0) \cdot T_{L-1}.
		\end{equation*}
		We have spent $2^{L}$ steps in visiting this subgraph. Observe that we also covered the edges between strings $s,r,$ which have  $(s_1,r_1) \in \left \{(0,1),(1,0)\right\}$. In particular, we will show that they are covered by the time $2^{L}/2$ on average, which intuitively can be seen by observing that we visit $G_{L-1}$ in $2^{L}$ steps, and an average edge will be visited in half that time. This gives us a cost 
		\begin{equation*}
			(\nu(0,1) + \nu(1,0)) \cdot  2^{L-1}.
		\end{equation*}
		Finally, the subgraph with $s,r$ such that $s_1=r_1=1$	is also isomorphic to $G_{L-1}$, and once again use the ordering $\sigma_{L-1}$ to traverse it in $2^{L}$ steps. In this case, we have a delay of $2^{L}$ due to visiting vertices with $r_1=0$ or $s_1=0$, and hence the edges are covered by the time  $2^{L} +  T_{L-1}$, and their total cost is 
		\begin{equation*}
			\nu(1,1) \cdot (2^{L}+T_{L-1}).
		\end{equation*}
		Hence, we have that 
		\begin{equation*}
			T_L \leq 		\nu(0,0) \cdot T_{L-1} + (\nu(0,1) + \nu(1,0)) \cdot  2^{L-1} + \nu(1,1) \cdot (2^{L}+T_{L-1}). 
		\end{equation*}
		Letting $t_L = T_L/2^{L+1}$ and replacing the values of $\nu$ yields
		\begin{equation*}
			t_L \leq  \frac{1+\rho}{4} t_{L-1} + \frac{1}{4},
		\end{equation*}
			which is a recurrence relation, and solving it shows that $t_L \to \frac{1}{3-\rho}$, regardless of $t_1$. Hence, for sufficiently large $L$ we should expect to get MSVC close to $\frac{1}{3-\rho}$.
    \par 
		With this intuition in mind, we now state and prove the theorem which gives the hardness reduction from Affine Unique Games to weighted graphs.
\begin{theorem}\label{hardness_min_sum_vc}
	For any $\varepsilon > 0, \rho \in (-1,0), \gamma > 0$, there is a sufficiently large alphabet size $L \in \mathbb{N}$ and a reduction from regular Affine Unique Games instances $\Lambda=(\mathcal{U},\mathcal{V},\mathcal{E},\Pi,[L])$ to weighted multigraphs $G = (V, E)$ with the following properties:
  \begin{itemize}
\item \emph{Completeness:} If $\Opt(\Lambda) \geq 1-\gamma$, then $\MSVC(G) \leq \left(\frac{1}{3-\rho} + \varepsilon+3\gamma\right) \cdot |V| \cdot W(E)$.
    \item \emph{Soundness:} If $\Opt(\Lambda) \leq \gamma$, then for every $r \in [0,1]$, $G$ is $(r, \Gamma_{\rho}(r)-\varepsilon)$-dense.
  \end{itemize}
	Moreover, the running time of the reduction is polynomial in $|\mathcal{U}|,|\mathcal{V}|,|\mathcal{E}|$, and exponential in $L$. The weights of edges in $G$ belong to the set $\left \{ \left(\frac{1+\rho}{4}\right)^{i} \left(\frac{1-\rho}{4}\right)^{L-i}\right\}_{i=0}^{L}$. The size of $|V|$  is at least $2^{L}$, and the total weights of edges is $W(E)=|V|\cdot D^{2}/2^L$, where $D$ denotes the degree of the regular Unique Games instance. Finally, the output graph $G$ is also regular, in the sense that the value $W(u,N(u))$ is uniform across all $u \in V$, and it equals $D^{2}2^{-L+1}$.
	\begin{proof}
    Let $\nu\colon \{0,1\}^2 \rightarrow [0,1]$ be the probability distribution over correlated uniformly distributed bits with negative correlation coefficient $\rho<0$.  In other words, we have
    \begin{equation*} 
			\nu(0,0)=\nu(1,1) = \frac{1+\rho}{4}, \quad \nu(0,1) = \nu(1,0) = \frac{1-\rho}{4}.
    \end{equation*}
    \par 
    Let us now describe how the multigraph $G$ can be constructed from $\Lambda$. We define the vertex set of $G$ to be $V = \mathcal{V} \times \{0,1\}^L = \{(v,x) \mid v \in \mathcal{V}, x \in \{0,1\}^L\}$. In particular, for every vertex $v \in \mathcal{V}$ we create $2^L$ vertices of $G$, which we identify with $L$-bit strings in $\{0,1\}^L$. We also write $v^x$ for a vertex $(v,x)$ of the graph $G$.     The edges of $G$ are constructed in the following way. For every $u\in \mathcal{U}$, and for every two $v_1,v_2\in N(u)$, we create an edge between vertices $v_1^x, v_2^y$ with weight
    \begin{equation*}
       \nu^{\otimes L}(x \circ \pi_{e_1}, y \circ \pi_{e_2}), \quad  \textrm{where } e_1 = (u,v_1), \quad e_2 = (u,v_2).
    \end{equation*}
    Expressed formally, the edge set $E$ is
    \begin{equation*}
      E = \{(e_1^x,e_2^y) \mid e_1 = (u,v_1), e_2 = (u,v_2), u \in \mathcal{U}, v_1,v_2 \in \mathcal{V}, x,y \in \{ 0, 1 \}^L \}.
    \end{equation*}
The number of vertices in $G$ is $|\mathcal{V}|2^L$, and the number of edges is $|\mathcal{V}|D^22^{L+1}$, so the construction is indeed polynomial in $|\mathcal{U}|,|\mathcal{V}|$ and $|\mathcal{E}|$, and exponential in $L$. Also, since $\mathcal{V} \neq \emptyset$ we have $|V| \geq 2^{L}$, and the weights of the edges indeed belong to the set specified in the statement of the theorem. Finally, the total weight of edges incident upon each vertex $v^{x}$ is the same for any $v^{x}$, and since $W_G(E)=D^{2} |\mathcal{V}|$, we have that $W_G(v^{x},N(v^{x})) = 2 D^{2} |\mathcal{V}| \frac{1}{|V|}= D^{2}2^{-L+1}$ for all $v^{x} \in V$.
    \par
		We are using the same reduction\footnote{This is the same as the Max-Cut hardness reduction in \cite{DBLP:journals/siamcomp/KhotKMO07}. Same reduction and soundness result also appeared in \cite{DBLP:journals/toc/AustrinKS11}, albeit with biased bits.} as the one used in Theorem 3.1. from \cite{DBLP:conf/approx/AustrinS19}, and the only difference is that we are starting from Affine Unique Games instead of (general) Unique Games. Since we are using the same reduction and Affine Unique Games are subsumed by the Unique Games, our graph $G$ satisfies the same soundness property as the one expressed by Theorem 3.1. in \cite{DBLP:conf/approx/AustrinS19}, and this is exactly the soundness property stated above. Hence, we only need to show completeness. 
		\par 
		For the completeness case let us assume  $\Opt(\Lambda) \geq 1-\gamma$. Therefore, there is a labelling $z \colon \mathcal{U} \sqcup \mathcal{V} \to \mathbb{Z}_L $ such that $\Val_z(\Lambda)\geq 1-\gamma$. In particular, there is $ \mathcal{\hat{E}} \subseteq \mathcal{E}, |\mathcal{\hat{E}}| \geq (1-\gamma) |\mathcal{E}|$, such that for each $e=(u,v) \in \mathcal{\hat{E}}$ we have $z(u)-z(v)=c_e$. Let us use $\hat{E} \subseteq E$ to denote the set 
		\begin{equation*}
			\hat{E} :=\{(e_1^x,e_2^y) \in E \mid e_1,e_2 \in \mathcal{\hat{E}} \}.
		\end{equation*}
Observe that $|\hat{E}| \geq |E|\cdot (1-2\gamma)$. Since the complement of $\hat{E}$ is of small fractional size, i.e., smaller than  $2 \gamma$, in the analysis we will focus on cover times of edges in $\hat{E}$, and we will trivially upper bound the cover time of edges in $\hat{E}^{c}$ by $|V|$. In particular, let us denote with $\hat{G}$ the graph $\hat{G} = (V,\hat{E})$ and find an ordering $\sigma$ such that $SVC_{\hat{G}}(\sigma) \leq (\frac{1}{3-\rho}+\varepsilon+\gamma)\cdot |V| \cdot W(E)$. As discussed this would then give us the stated completeness
		\begin{equation*}
SVC_G(\sigma) \leq \left(\frac{1}{3-\rho}+\varepsilon+3\gamma\right)\cdot |V| \cdot W(E), 
		\end{equation*}
		by bounding the cover time of edges in $\hat{E}^{c}$ by $|V|$.
		\par 
		Before explaining how $\sigma$ is constructed, let us first introduce some notation. We use  $z_1,\hdots,z_{L}\colon  \mathcal{V} \to \mathbb{Z}_L$ to denote the mappings defined by 
		\begin{equation*}
			z_i(v) \colon = z(v) +i, \quad \textrm{for } i \in [L].
		\end{equation*}
		Let us then define sets $F_i^0,F_i^1 \subseteq V$, as the sets in which, for every $v \in \mathcal{V}$, inside the long code $(v,x)$ we fix the $z_i(v)$-th coordinate to $0$ or $1$, respectively. In particular, we have 
		\begin{equation*}
			F_i^0 = \left \{(v,x) \in V \mid x_{z_i(v)} =0 \right\}, \quad 			F_i^1 = \left \{(v,x) \in V \mid x_{z_i(v)} =1 \right\}.
		\end{equation*}
Intuitively, the sets $F_i^0$	 (or $F_i^1$) for a fixed $i$ fix the values at the coordinates in which labels ``agree''. Then, we use the sets $F_i^0$ and  $F_i^1$ to construct ordering inductively. First, we define the ordering on $C_{L-1}=F_1^0 \cap F_2^0 \cap \hdots F_{L-1}^0$, then using this ordering we define ordering on $C_{L-2}=F_1^0 \cap F_2^0 \cap \hdots F_{L-2}^0$, and so on until we construct an ordering on $C_1=F_1^0$ and finally on $C_0$ which we define to be $C_0:=V$. As we are defining orderings on $C_i,i=0,\hdots,L-1$, we will be expressing an upper bound $T_i$ for the average time edges $E_i$ with both endpoints endpoints in $C_i$ are covered by the ordering. Let us use $G_i=(C_i,E_i)$ to denote graphs these edges belong to. Before discussing our ordering, let us make an observation that $|C_{i}| =  2 \cdot |C_{i+1}|$, since $C_{i}$ has one more free coordinate for each $v \in \mathcal{V}$.
		\par 
We discuss the ordering for $C_{L-1}$ first. Before that, let us remark that the particular ordering and the cost of covering edges in $C_{L-1}$ will be inconsequential for the final value that we get in this theorem. The main reason we discuss this case here is because we believe it will be a good preparation for discussing the inductive step that will follow. In the first step, we iterate through $v \in \mathcal{V}$ in a random order\footnote{As we have said, the value obtained in the first step is not relevant as it will be seen later. Hence, we can also choose to visit $v \in \mathcal{V}$ in any fixed order in which case we can also use a trivial upper bound of $|V_{L-1}|$ on $T_{L-1}$.}, and pick $(v,x) \in C_{L-1}$ such that\footnote{Due to symmetry it is not important whether we pick  $x_{z_L(v)}=0$ or $x_{z_L(v)}=1$ in the first iteration, as long as we keep that choice fixed.} $x_{z_L(v)}=0$. Then, we iterate through $v \in \mathcal{V}$ in a random order and pick the remaining vertex at each $(v,x)$, i.e., the vertex with $x_{z_L(v)}=1$. Let us upper bound the average time an edge $e \in E_{L-1}$ with both endpoints in $C_{L-1}$ is visited by this schedule. Observe that we spent $\frac{1}{2}|C_{L-1}|$ time in the first step, and $\frac{1}{2}|C_{L-1}|$ in the second step. Thus, if an edge with both endpoints in $C_{L-1}$ has at least one endpoint with a label $0$ at $x_{z_L(v)}$, then this point will be picked in the first step on average by the time  $\frac{1}{4}|C_{L-1}|$. Otherwise, if the edge $e$ has both endpoints $v_1^{x},v_2^{y}$ picked in the second step, i.e. $x_{Z_L(v_1)}=1, y_{Z_L(v_2)}=1$, then it will be picked on average by the time $\frac{3}{4} |C_{L-1}|$. Since the weight of edges from $E_{L-1}$ picked in the first step is $(\nu(0,0)+\nu(0,1)+\nu(1,0))\cdot W_{G_{L-1}}(E_{L-1}) = \frac{3-\rho}{4} W_{G_{L-1}}(E_{L-1})$, and the weight of the remaining edges that we consider is $\nu(1,1) \cdot W_{G_{L-1}}(E_{L_1})=\frac{1+\rho}{4} \cdot W_{G_{L-1}}(E_{L_1})$, the average cover time is 
		\begin{equation*}
			T_{L-1} = \frac{3-\rho}{4} \cdot \frac{1}{4}|C_{L-1}|+ \frac{1+\rho}{4} \frac{3}{4}|C_{L-1}| = \frac{3 +\rho}{8} |C_{L-1}|.
		\end{equation*}
		We observe that this also shows that there is an ordering $\sigma_{L-1}$ which covers an edge in $E_{L-1}$ on average by the time $T_{L-1}$.
		\par 
Let us now fix $i=0,\hdots,L-2$, and assume that we have an ordering $\sigma_{i+1}$ of vertices in $C_{i+1}$ such that the edges in $E_{i+1}$ are covered by the time $T_{i+1}$ on average, and let us use this procedure to construct an ordering of the vertices in $C_{i}$ and derive a suitable upper bound on $T_{i}$. We assume that the weights of the edges $E_{i}$ are normalized so that they sum up to $1$. The ordering in $C_{i}$ works as follows. First, using $\sigma_{i+1}$ we visit vertices in $C_{i+1} = C_{i} \cap F_i^0$. The total weight of the edges with both endpoints in $C_{i} \cap F_i^0$ is $\nu(0,0) \cdot W_{G_i}(E_i)$, and they are covered by  $\sigma_{i+1}$ until $T_i$ on average. Hence, the cost for these edges is 
		\begin{equation}\label{cost1}
T_{i+1} \cdot \nu(0,0) \cdot W_{G_i}(E_i).	
		\end{equation}
Furthermore, during this pass, we have also visited all the edges with one endpoint in $C_{i} \cap F_i^0$ and another endpoint in $C_{i} \cap F_i^1$, and their total weight is $(\nu(0,1)+\nu(1,0))\cdot W_{G_i}(E_i)$. Also, these covers are disjoint (each one of these edges will be visited only once in the first pass). Since the starting Unique Games instance was regular and we removed at most $2 \gamma$ edges, the edges on average will be covered by the time 
		\begin{equation}
			\frac{1+2\gamma}{2}|C_{i+1}|
		\end{equation}
		at most. Hence, the cost for these edges is 
		\begin{equation}\label{cost2}
(\nu(0,1)+\nu(1,0))\frac{1+2\gamma}{2} |C_{i+1}| \cdot W_{G_i}(E_i).
		\end{equation}
		Finally, we pass through the vertices in $C_{i} \cap F_i^1$. The graph induced by this vertex set is actually isomorphic to $C_{i+1} = C_{i}\cap F_i^{0}$, and hence we can once again use the ordering $\sigma_{i+1}$. Then, the edges in this graph are visited on average by the time
		\begin{equation*}
			|C_{i+1}| + T_{i+1}, 
		\end{equation*}
		where the $|C_{i+1}|$	term is due to the delay coming from the first pass. Hence, the cost of these edges is at most 
		\begin{equation}\label{cost3}
\nu(1,1)( |C_{i+1}| + T_{i+1}) \cdot W_{G_i}(E_i).
		\end{equation}
		Adding up \eqref{cost1},\eqref{cost2} and \eqref{cost3} we get that
		\begin{equation}\label{T_i_recursive}
			T_i \leq \frac{1+\rho}{4} T_{i+1} + \frac{1-\rho}{2} \frac{1+2\gamma}{2}|C_{i+1}| + \frac{1+\rho}{4} \cdot (|C_{i+1}| + T_{i+1}).
		\end{equation}
		If we let $t_i = T_i/|C_i|$ and divide both sides by $|C_i|=2|C_{i+1}|$, we can write \eqref{T_i_recursive} as
		 \begin{equation*}
			 t_i \leq \frac{1+\rho}{8}	t_{i+1} + \frac{1-\rho}{4}\frac{1+2 \gamma}{2}  + \frac{1+\rho}{4}\left( \frac{1+t_{i+1}}{2}\right),
		\end{equation*}
		which can be simplified to 
		\begin{equation}\label{t_i_eq}
			t_i \leq \frac{1}{4} + \frac{1+\rho}{4} t_{i+1} + \frac{1-\rho}{4}\gamma.
		\end{equation}
		Let us show that $t_i \leq  \frac{1}{3-\rho}+\gamma +2^{-L+i}$ as follows. Let us define $r_i = t_i -\gamma - \frac{1}{3-\rho}$. By substituting $t_i = \frac{1}{3-\rho}+\gamma+ r_i$ into \eqref{t_i_eq} we obtain
		\begin{equation}
			\frac{\gamma}{2}+r_i \leq  \frac{1+\rho}{4}r_{i+1}.
		\end{equation}
		Since $ \rho \in (-1,0)$ and $\gamma>0$ we have
		\begin{equation}\label{r_i_eq}
			r_i \leq  \frac{1}{2}  r_{i+1}.
		\end{equation}
		Hence, since by calculation for $T_{L-1}$ we have $r_{L-1} \leq \frac{1}{2}$, which with \eqref{r_i_eq} implies that $r_i \leq 2^{-L+i}$, and therefore $t_0 \leq \frac{1}{3-\rho} + \gamma+2^{-L}$. By letting $L$ be large enough so that $2^{-L} \leq \varepsilon$ and recalling that $t_0 = T_0 / |V| $ we get
		\begin{equation*}
			T_0 \leq 	(\frac{1}{3-\rho} + \gamma+\varepsilon) |V|.
		\end{equation*}
Recalling that we used $T_0$ to denote an average time, we have that the actual time is at most $T_0 \cdot W(E)= (\frac{1}{3-\rho} + \gamma+\varepsilon) |V|W(E)$, 
	\end{proof}
\end{theorem}
This reduction outputs a weighted graph. In the Section \ref{section:remove_weights} we will show how this weighted graph can be transformed into an unweighted graph with \emph{essentially} the same properties using a polynomial time reduction. For now, let us briefly discuss how soundness and completeness properties stated in the theorem above are useful for studying Min Sum Vertex Cover. 
\par
For the completeness, we will get that $\MSVC(G) \leq \left(\frac{1}{3-\rho} + \varepsilon+3\gamma\right)\cdot |V| \cdot W(E)$. On the other hand, in the soundness case we have that for any ordering $\sigma$ we have 
\begin{equation*}
	\begin{split}
\SVC_G(\sigma) = \sum_{t=1}^{|V|} W(\sigma([t])^c,\sigma([t])^c) \geq  \sum_{t= 1 }^{n} \Gamma_{\rho}(1-t/n)\cdot W(E) -\varepsilon \cdot W(E) \\
		= \left(\int_{0}^{1} \Gamma_{\rho}(1-r) dr -\varepsilon\right) \cdot |V|\cdot W(E)+ O(W(E)) \\
		= \left(\int_{0}^{1} \Gamma_{\rho}(r) dr -\varepsilon\right) \cdot |V| \cdot W(E) +O(W(E)) .
	\end{split}
\end{equation*}
Hence, by letting $\gamma \to 0, \varepsilon \to 0, |V| \to \infty,$ we get an inapproximability ratio of
\begin{equation*}
	\frac{\int_{0}^{1} \Gamma_{\rho}(r) dr}{\frac{1}{3-\rho}}.
\end{equation*}
This expression is minimized for $\rho\approx -0.52$, for which the inapproximability ratio is approximately $1.0157$. \par 
\subsection{Improved Jardness for Non-Regular Graphs}\label{subsection:non_regular_hardness}
In this section we describe how we can use the properties of regular graphs output by the reduction from Section \ref{section:UG_reduction} to obtain the hardness of approximating MSVC within $1.0748$.
\par 
Our main idea is simple. Fix	$k \in \mathbb{N}$, $k$  weights $\alpha_1,\alpha_2,\hdots,\alpha_k \in \mathbb{R}_+$, and $k$  correlation parameters $\rho_1,\hdots,\rho_k \in (-1,0]$. We then construct an instance $\bar{G}$ which consists of $k$  disjoint graphs $G_1,\hdots,G_k,$ in which the edges of $G_i$ are obtained from graphs $G$ introduced in Theorem \ref{hardness_min_sum_vc} with the correlation parameters $\rho_i$, respectively. The weights of edges in $G_i$ are multiplied by $\alpha_i$.
\par 
Let us now discuss the soundness and the completeness properties of the newly created graph $\bar{G}$. In both soundness and completeness case at each step $t$, the optimal ordering will choose an unpicked vertex $v$  from one of the graphs $G_1,\hdots,G_k$. Furthermore, if the vertex $v$ is picked from the graph $G_i$, it will be exactly the next unpicked vertex in the optimal ordering for the graph $G_i$. Also, the optimal choice of the graph $G_i$  will be exactly the graph in which the picked vertex $i$  covers the largest weight of uncovered edges. In other words, the choice of the graph $G_i$  is greedy.
\par 
In order to calculate which graph to pick and how many edges are covered, for each graph $G_i$  we require a function $c^i(t)$ which gives a lower bound on the relative weight of edges covered by the fractional time $t$, and the function $s^i(t)$  which gives an upper bound on the relative weight of the edges covered by the fractional time $t$.
\par 
Let us discuss how to express the functions $s^i$  and $c^i$  now. For the sake of notational simplicity we drop the superscript in the discussion that follows, and assume that $\alpha_i=1$.
\par
In the soundness case, we have that $s(t) = 1-\Gamma_\rho(1-t) + \varepsilon$, since after fractional time $t$  at least $\Gamma_\rho(1-t)-\varepsilon$ fraction of edges is not covered yet, and hence at most $1-\Gamma_\rho(1-t)+\varepsilon$ fraction of edges can be covered.
\par 
In the completeness case, the function $c$  actually depends on the parameter $L$ which is the alphabet size in Theorem \ref{regular_main_theorem}. Hence, let us define $c_i$  to be the function $c$  for the alphabet size $i$. For $c_1(t)$, we have that $c_1(0)=0$, for $t\in (0,1/2]$ we have that
\begin{equation*}
	c_1(t)\geq \left[\nu(0,1)+\nu(1,0)\right] \cdot 2\cdot t - 2 \gamma,
\end{equation*}
where the right hand side can be explained as follows. In the first $1/2$  steps we disjointly cover good edges with exactly one endpoint in $F_1^0$, and we have the term $[\nu(0,1)+\nu(1,0)] \cdot 2t$ due to that. The term $2 \gamma$ is due to the cover time of bad edges. Finally, we have that 
\begin{equation*}
	c_1(t) \geq \nu(0,0)+\nu(0,1)+\nu(1,0) - 2 \gamma,
\end{equation*}
for $t \in (1/2,1]$ since we covered edges of weight $\nu(0,0)+\nu(0,1)+\nu(1,0)$ by the time $t=1/2$.
\par 
We derive a lower bound on $c_i$  inductively. In particular, for $t \in [0,1/2]$ we have 
\begin{equation*}
c_{i+1}(t) \geq \nu(0,0) \cdot c_i(2\cdot t)  + (\nu(0,1)+\nu(1,0))\cdot 2 t - 2 \gamma.
\end{equation*}
Once again, the term $2 \gamma$ is due the bad edges. The term $\nu(0,0)\cdot c(2\cdot t)$ is because the edges with both endpoints in $F_1^0$  will follow the visit time in the graph $F_1^0$ which is isomorphic to the graph  constructed with $i$  labels. Finally, the term $(\nu(0,1)+\nu(1,0))\cdot t$  exists due to disjoint cover of the edges with only one endpoint in $F_1^0$. \par
For $t \in (1/2,1)$, we have that 
\begin{equation*}
	c_{i+1}(t) \geq  [\nu(0,0)+\nu(1,0)+\nu(0,1) ] + c_i( 2\cdot t -1 ) \cdot \nu(1,1) - 2 \gamma,
\end{equation*}
The term $[\nu(0,0)+\nu(1,0)+\nu(0,1)]$ exists since we covered that fraction of good edges in the first half of the ordering, and the term $c_i( 2\cdot t -1 ) \cdot \nu(1,1) $ exists since the graph $F_1^1$  is isomorphic to the graph created  with $i$ labels, while the term $2 \gamma$ is due to the existence of the bad edges.
\par
Since we can choose arbitrarily big alphabet size $L$, instead of a particular function $c_i$  we can study the function $c$  which is obtained by letting $i \to \infty$. While we do not have a closed form expression for $c$, we can calculate it to high precision by iterating recurrence formula given above.
\par 
This allows us to get inapproximability ratio for any fixed set of weights $\alpha_1,\hdots,\alpha_k$. By using grid search paired with a gradient descent we can find values of $\alpha_i$  for which inapproximability ratio is $1.0748$. The values of $\alpha_i$ and $\rho_i$ are given in Figure \ref{rho_alpha_table}. The author also provided the code which verifies the claims in this section (and thus with the Section \ref{section:remove_weights} proves Therem \ref{main_theorem}) on his website.
\begin{figure}
	\begin{center}
		\begin{tabular}{ |c| c| c| c| c| c| c|c|c|c| c| c|c|  }
			\hline
			 i & 1 & 2 & 3 & 4 & 5 & 6 & 7 & 8 & 9 & 10\\ 
				$\alpha_i$ & 1 &40.20 &40.22 &43.78 &43.78 &43.81 &43.82 &47.74 &47.81 &53.19 \\
				 $\rho_i$ &-0.979 & -0.974 & -0.975 & -0.968 & -0.970 & -0.972 & -0.973 & -0.962 & -0.964 & -0.929 \\
			\hline 
			\hline
				i & 11 & 12 & 13 & 14 & 15 & 16 & 17 & 18 & 19 & 20\\ 
				$\alpha_i$& 53.50 &53.53 &53.76 &54.10 &58.75 &59.14 &63.07 &63.09 &67.61 &67.64 \\
			$\rho_i$ &-0.931 & -0.936 & -0.937 & -0.970 & -0.921 & -0.923 & -0.912 & -0.914 & -0.903 & -0.905  \\
			\hline
			\hline
				i & 21 & 22 & 23 & 24 & 25 & 26 & 27 & 28 & 29 & 30\\ 
				$\alpha_i$& 72.36 & 72.41 & 79.84 & 79.94 & 81.10 & 81.21 & 94.5 & 94.67 & 96.14 & 96.27  \\
			$\rho_i$ &-0.894 & -0.896 & -0.876 & -0.879 & -0.882 & -0.884 & -0.856 & -0.859 & -0.862 & -0.866 \\
			\hline
			\hline
				i & 31 & 32 & 33 & 34 & 35 & 36 & 37 & 38 & 39 & 40\\ 
				$\alpha_i$ & 98.32 & 98.5 & 118.05 & 118.25 & 120.76 & 121 & 123.61 & 123.94 & 145.96 & 146.39 \\
			$\rho_i$ &-0.857 & -0.860 & -0.840 & -0.844 & -0.841 & -0.846 & -0.841 & -0.844 & -0.837 & -0.840 \\
			\hline
			\hline
				i & 41 & 42 & 43 & 44 & 45 & 46 & 47 & 48 & 49 & 50\\ 
				$\alpha_i$ & 150.07 & 150.58 & 169.25 & 169.78 & 186.55 & 187.22 & 214.35 & 217.53 & 222.30 & 222.94 \\
			$\rho_i$ & -0.838 & -0.840 & -0.837 & -0.842 & -0.841 & -0.844 & -0.839 & -0.845 & -0.847 & -0.851 \\
			\hline
			\hline
				i & 51 & 52 & 53 & 54 & 55 & 56 & 57 & 58 & 59 & 60\\ 
				$\alpha_i$ & 260.16 & 265.34 & 299.47 & 306.72 & 353.2 & 361.89 & 436.79 & 448.42 & 607.90 & 608.43 \\
			$\rho_i$ & -0.853 & -0.858 & -0.865 & -0.865 & -0.876 & -0.878 & -0.891 & -0.895 & -0.916 & -0.917 \\
			\hline
		\end{tabular}

	\end{center}
\caption{Weights $\alpha_i$  and correlation parameters $\rho_i$  used to get graph $G=G_1\cup \hdots \cup G_k$ for which the MSVC is hard to approximate within $1.0748$.}
\label{rho_alpha_table}
\end{figure}

\section{Removing Weights and the Proof of the Main Lemma} \label{section:remove_weights}
In this section we discuss how the weighted graphs constructed in the previous section can be reduced to unweighted graphs with similar properties.
We sketch the idea first, before giving the details.  \par The reduction consists of two steps. The first step will take a weighted graph\footnote{This can be either a graph from Section \ref{section:UG_reduction} or Serction \ref{subsection:non_regular_hardness}.} $G$, and construct another weighted graph $G'$ by replacing every vertex $v$ from $G$ by $m$ vertices in $G'$, replacing each $e=(u,v)$ of weight $W_e$ by a complete bipartite graph with the weights of $m^2$ edges being  $W_e$. We will show that this graph satisfies quantitatively the same completeness and soundness properties as the graph $G$. 
\par 
Then, to remove the weights, for each complete bipartite graph between $u$ and $v$, we replace $m^{2}$ edges with approximately $m^{2}W_e$ edges, in such a way that the bipartite graph has good regularity properties in the sense of Szemerédi. This procedure outputs an unweighted graph with the right density in each bipartite graph which ``simulates'' edge $e$. This will give us an unweighted graph $G''$. If we use the graph from Section \ref{section:UG_reduction} as an input to the reduction, we will show that we obtain a regular graph $G''$, and furthermore this graph will satisfy properties sufficient for proving Theorem \ref{regular_main_theorem}. If we start with the graph from Section \ref{subsection:non_regular_hardness}, we will obtain the graph $G''$ which will not be regular, but its properties will be sufficient to prove Theorem \ref{main_theorem}.
\par 
Let us give the details now. We begin by showing that $G'$ satisfies similar completeness and soundness properties as the graph $G$.
\begin{lemma} \label{lemma_G'}
	Consider a weighted graph $G=(V,E)$, and let $m \in \mathbb{N}$. Let $G'=(V',E')$ be a graph obtained from $G$ by replacing each vertex $v$ in $G$ by $m$ vertices  $v_1,\hdots,v_m,$ and by replacing each edge $e=(u,v)$ of $G$ by $m^{2}$ edges of weight $W_e$. Then, the graph $G'$ satisfies the following properties:
	\begin{itemize}
\item \emph{Completeness:} If $\MSVC(G) \leq  \tau |V| W_G(E)$ for some $\tau \in \mathbb{R}_+$ , then $\MSVC(G') \leq \tau|V'| W_{G'}(E')$.
\item \emph{Soundness:} In case $G$ is $(r, g(r))$-dense for some continuously differentiable function $g\colon [0,1] \to [0,1]$ with $|g'|\leq 1$ , then $G'$ is $(r, g(r)-1/|V|)$-dense.
	\end{itemize}
	\begin{proof}
We start by proving the completeness property. Hence, let us fix $\sigma$ such that $\SVC_{\sigma}(G) \leq \tau |V| W_G(E)$. Then, let us consider $\sigma'$ which first visits duplicates of $\sigma(1)=v$, i.e., it visits first $v_1,\hdots,v_m$, after which it visits all the duplicates of $\sigma(2)$, and so on until it visits duplicates of $\sigma(n)$. Observe that the order of visiting $v_i,i=1,\hdots,m,$ is irrelevant, for each $v \in V$. Consider any edge $e \in E$ covered by the vertex  $v = \sigma(t)$ at the time $t$ by the schedule  $\sigma$. Then, the vertices $v_1,\hdots,v_m$ will be picked in times $(t-1)\cdot m +1, \hdots, t \cdot m$. Since the  edge $e$ was replaced by the complete bipartite graph, edges in this graph will be covered by the time $t\cdot m + (1-m)/2$ on average. Therefore, we have that 
		\begin{equation*}
			\begin{split}
				\MSVC(G') \leq \sum_{e \in E} (c_{\sigma,e} \cdot m + \frac{1-m}{2}) \cdot W_e \cdot m^{2} \leq   \sum_{e \in E} c_{\sigma,e} W_e \cdot m^{3} .
			\end{split}
		\end{equation*}
		Now, we have that $ \sum_{e \in E} c_{\sigma,e} W_e = \MSVC(G)$, and therefore we have 
		\begin{equation*}
			\MSVC(G')  \leq m^{3} \MSVC(G).
		\end{equation*}
		Finally, we have that $\MSVC(G) \leq \tau |V| W_{G}(E)$, and hence
		\begin{equation*}
			\MSVC(G') \leq \tau |V| W_{G}(E) m^{3}.
		\end{equation*}
Since $|V'|=m|V|$	and $W_{G'}(E')=m^{2} W_G(E)$  we can rewrite this as 
		\begin{equation*}
			\MSVC(G') \leq \tau |V'| W_{G'}(E'),
		\end{equation*}
		which is what we claimed by the completeness property. \par
Let us now show the soundness property.	Hence, let us assume that for every $r \in [0,1]$, $G$ is $(r, g(r))$-dense, pick an arbitrary $S' \subseteq V', |S'|=r|V'|$, and show that $w(S',S') \geq g(r)-1/|V|$. For a set $S'$, let us define $\left \{q_v\right\}_{v \in V}$ to be the quantities which measure for each $v \in V$  how many vertices from $S'$ are in $A_v :=\left \{v_1,\hdots,v_m\right\}$. In particular, for each $v \in V$ we let 
		\begin{equation*}
			q_v = \left| S' \cap \left \{v_1,\hdots,v_m\right\} \right| / m.
		\end{equation*}
		For $T' \subseteq V'$, let us use $f(T')$  to denote the number of $v \in V$  such that $\left|A_v \cap T'\right| \not \in  \left \{0,m\right\}$. Then, we consider the following two cases:
		\begin{itemize}
			\item[(a)]  $f(S') \leq 1$.
			\item[(b)] $f(S') \geq 2$.
		\end{itemize}
Let us first resolve case (a). Let $S=\left \{v \in V \mid q_v = 1\right\}$. Then we have that $|S| \geq r|V| -1$. Furthermore, for each edge $e$  in $N(S,S) \subseteq E$  of weight $W_e$   we have $m^{2}$  edges in $N(S',S') \subseteq E'	$ of total weight $W_e m^{2}$. Hence, $w(S',S') \geq	w(S,S)$. Finally, since $S \subseteq V$  and graph $G$  satisfies the soundness property, we have that $w(S,S) \geq g(r-1/|V|)$. Hence, we have that 
		\begin{equation*}
w(S',S')  \geq 	g(r-1/|V|).
		\end{equation*}
		By the mean value theorem, for any $x,y \in [0,1], x<y$, we have
		\begin{equation*}
			|g(x) - g(y)| \leq \sup_{t \in [x,y]} g'(t) (y-x).
		\end{equation*}
		 Since $ \left|g'(z)\right| \leq 1$ for all $z \in [0,1]$, by taking $x=r-1/|V|, y = r$   we get $\left|g(r-1/|V|) -  g(r) \right| \leq 1/|V|$, and therefore
		\begin{equation*}
			|w(S',S') | \geq 	g(r)-1/|V|.
		\end{equation*}
		Let us now resolve case (b). Our proof strategy is as follows. We first fix two vertices, $u,v$ such that $0<q_u,q_v<1$. Then, we consider two possible actions:
		\begin{itemize}
			\item \emph{A1:} Moving as many vertices from $A_u\cap S'$ to $A_v\cap S'$ as possible. This would either make $q_u=0$ or $q_v=1$. In particular, we would move $\min(q_u,1-q_v)\cdot m$  vertices from $A_u$ to $A_v$.
			\item \emph{A2:} Moving as many vertices from $A_v \cap S'$ to $A_u \cap S'$ as possible. This would either make $q_v=0$  or $q_u=1$. In particular, we would move $\min(q_v,1-q_u)\cdot m$  vertices from $A_v$ to $A_u$. 
		\end{itemize}
		First and the second operation will create vertex sets $S'_1,S'_2$, respectively, for which $f(S'_1) < f(S)$ and $f(S'_2) <f(S)$.  We will furthermore show that we have 
		\begin{equation*}
			\left|	w(S_1',S_1') \right| \leq  \left|w(S',S') \right| \quad \textrm{ or } \quad  \left| w(S_2',S_2') \right| \leq \left| w(S',S') \right|.
		\end{equation*}
		We then pick $T_1$   to be $S_i'$  with the smaller $\left|w(S_i',S_i')\right|$. Then we have that $f(T_1) < f(S) $, $\left|w(T_1,T_1) \right| \leq  \left|w(S',S')\right|$. Repeating this procedure finite number of steps, we will get $T_j \subseteq  V',j \in \mathbb{N}$, such that 
		\begin{equation*}
			f(T_j) \in \left \{0,1\right\}, \quad |T_j| = r |V'|, \quad |w(T_j,T_j)| \leq  \left| w(S',S') \right|.
		\end{equation*}
		Then, by the point (a) we have that 
		\begin{equation*}
			|w(T_j,T_j)| \geq  g(r) - 1/|V|,
		\end{equation*}
			and therefore 
			\begin{equation*}
				\left| w(S',S') \right| \geq g(r)-1/|V|.
			\end{equation*}
		It remains to show that one of the sets $S_1'$ or $S_2'$, created by actions \emph{A1} or \emph{A2}, respectively, will be such that 
		\begin{equation*}
			\left|	w(S_1',S_1') \right| \leq  \left|w(S',S') \right| \quad \textrm{ or } \quad  \left| w(S_2',S_2') \right| \leq \left| w(S',S') \right|.
		\end{equation*}
		We equivalently show that 
		\begin{equation*}
			\left|	W(S_1',S_1') \right| \leq  \left|W(S',S') \right| \quad \textrm{ or } \quad  \left| W(S_2',S_2') \right| \leq \left| W(S',S') \right|.
		\end{equation*}
		The main idea relies on the following observation. If we denote with $h$ the amount of vertices moved from $A_u$ to  $A_v$, where we allow $h<0$ to indicate moving vertices from  $A_v$ to $A_u$, then the difference in the number of edges obtained by moving $h$ vertices is a concave function, and hence it attains its minimum at the edge of an interval it is defined over. Let us give the details now. We use $W_e$ to denote the weight of an edge between $u$  and $v$. In case there is no edge between $u$  and $v$  we set $W_e =0$. Furthermore, let us define
		\begin{equation*}
			W_1 = \left| W(	A_u, S' \setminus A_v)\right| , \quad W_2 =\left| W(	A_v, S' \setminus A_u)\right|.
		\end{equation*}
		In particular, $W_1$ 	is the total weight of edges from $A_u$ to the rest of $S'$ excluding edges with endpoints in $A_v$, and $W_2$  is defined analogously for $A_v$. Then, moving $\min(q_u,1-q_v)\cdot m$ by action \emph{A1}  will create the following changes in $|W(S_1',S_1')|$  compared to $|W(S',S')|$:
		\begin{itemize}
			\item Removing edges between $A_u \cap S'$   and  $A_v \cap S'$, which eliminates edges of weight $W_e \cdot \min(q_u,1-q_v)\cdot q_v \cdot m^{2}$. 
			\item Removing edges between $A_u \cap S'$ and $S' \setminus A_v$, which eliminates edges of weight $W_1 \cdot \min(q_u,1-q_v)$.
			\item Adding edges between $A_v \cap S'$ and $S' \setminus A_u$, which adds edges of weight $W_2 \cdot \min(q_u,1-q_v)$.
		\end{itemize}
		Hence, we have that 
		\begin{equation*}
			\begin{split}
				|W(S_1',S_1')| = |W(S',S')| -W_e \cdot \min(q_u,1-q_v)\cdot q_v \cdot m^{2} - W_1 \cdot \min(q_u,1-q_v) +W_2 \cdot \min(q_u,1-q_v) \\ 
				= |W(S',S')| -W_e \cdot \min(q_u,1-q_v)\cdot q_v \cdot m^{2} +  \min(q_u,1-q_v) (-W_1 +W_2).
			\end{split}
		\end{equation*}
		Similarly, we get that 
		\begin{equation*}
			|W(S_2',S_2') | 	= |W(S',S')| -W_e \cdot \min(q_v,1-q_u)\cdot q_u \cdot m^{2}+  \min(q_v,1-q_u) (-W_2 +W_1).
		\end{equation*}
		Since either $W_1-W_2 \leq 0$  or $W_2 - W_1 \leq 0$  we have that either 
	\begin{equation*}
		|W(S_1',S_1')|  \leq |W(S',S')|,
	\end{equation*}
	or	
	\begin{equation*}
		|W(S_2',S_2')|  \leq |W(S',S')|,
	\end{equation*}
	which concludes our proof.
	\end{proof}		 
\end{lemma} 
In the next step of our reduction we replace each complete graph between $A_u,A_v,$ with a regular unweighted graph between the same vertices, by sampling the edges at an appropriate density. In particular, we show the following lemma.
\begin{lemma}\label{gadget_lemma}
	Let $q \in \mathbb{N}$, and let $0<\varepsilon<1/2, 0 < W_{\min} < W_{\max} < 1$. Furthermore, let $\varepsilon$ be rational number whose denominator divides $q$. Then there is sufficiently large $m \in \mathbb{N}$, such that the following holds. For each complete weighted bipartite graph	$H=(A_u \cup A_v,E_H)$ where $|A_u |= |A_v |= m$, in which each edge has weight $W_e$, and $W_e \in  [W_{\min},W_{\max}]$ is a rational number whose denominator divides $q$, there is an unweighted graph $H'=(A_u \cup A_v,E_{H'})$ such that for each  $S' \subseteq A_u \cup A_v$ we have 
	\begin{equation*}
		\left|W_{H}(S',S') - W_{H'}(S',S') \right| \leq 3 \varepsilon m^{2} W_e. 
	\end{equation*}
	Furthermore, each vertex in the graph $H'$ has degree $m W_e \cdot (1+\varepsilon)$. 
	\begin{proof}
		We prove this theorem using probabilistic method.	In particular, let us consider a random biparite graph $\bar{\mathcal{H}}$ on vertices $A_u \cup A_v$, obtained from $H$ by sampling each edge $e$ of $H$ of weight $W_e$ with probability $W_e$. Consider $S=S_u \cup S_v, S \subseteq A_u \cup A_v$, and let us use $Q(S)$ to denote the event
		\begin{equation*}
			Q(S):= \left|W_{\mathcal{H}}(S,S) - W_{H}(S,S) \right| > \varepsilon  W_e m^{2}.
		\end{equation*}
Let us enumerate weighted edges between $S_u$ and $S_v$  by $i=1,\hdots,|S_u|\cdot |S_v|$, and use $\{X_i\}_{i=1}^{|S_u|\cdot |S_v|}$  as the indicator functions for these edges in $\mathcal{H}$. Each $X_i$  is a random variable with $\E[X_i]=W_e$. All $X_i$ are independent, and $\E\left[\sum_{i=1}^{|S_u|\cdot |S_v|} X_i \right] = W_e |S_u|\cdot |S_v|= W_H(S,S) $. Also, $W_{\mathcal{H}}(S,S) = \sum_{i=1}^{|S_u|\cdot |S_v|} X_i$. We can then write
		\begin{equation*}
			\begin{split}
			\Pr[Q(S)] =\Pr\left[\left|W_{\mathcal{H}}(S,S) - W_{H}(S,S) \right| > \varepsilon W_e m^{2}\right] 
			= \Pr\left[ \left| \sum_{i=1}^{|S_u|\cdot |S_v|} X_i -  W_e |S_u|\cdot |S_v|  \right|  \geq \varepsilon W_e m^{2} \right] \\
				\leq 2 e^{- \frac{W_e^{2} \varepsilon^{2} m^{4} }{|S_u|\cdot |S_v| } },
			\end{split}
		\end{equation*}
		where in the last inequality we used Hoeffding's inequality \eqref{hoef_bound_eq}. Now, since $|S_u|\cdot |S_v| \leq m^{2}$  we have that 
		\begin{equation*}
				\Pr[Q(S)] \leq 2 e^{- \varepsilon^{2} W_e^{2} m^{2} }.
		\end{equation*}
		Since there are $2^{2m}$ different subsets $S \subseteq A_u \cup A_v$, we have by the union bound that 
		\begin{equation}\label{hoefding_regularity}
			\Pr\left[ \left(\exists S \subseteq A_u \cup A_v \right)	Q(S)\right] \leq 2^{2 m} \cdot 2 e^{- W_e^{2} \varepsilon^{2} m^{2} }.
		\end{equation}
		Note that this probability goes to $0$ as $m\to \infty$. Let us also show that the degrees of vertices in $\mathcal{H}$  will be close to $W_e m$  with high probability. Fix any $z \in A_u \cup A_v$, and let us enumerate edges in $H$  with one endpoint in $z$   by $1,\hdots,m$, and indicator functions of those edges in $\mathcal{H}$  by $X_1,\hdots,X_m$. Let $Q(z)$   be the following event:
		\begin{equation*}
			Q(z):= \left|\sum_{i=1}^{m}X_i - W_e \cdot m \right| > \varepsilon W_e \cdot m.
		\end{equation*}
		Since $X_i$ are independent and $\E[\sum_{i=1}^{m}X_i ] = W_e \cdot m$ by the Hoeffding's inequality \eqref{hoef_bound_eq} we have that 
		\begin{equation*}
			\Pr[Q(z)] \leq 2 e^{-\frac{\varepsilon^{2} W_e ^{2} \cdot m^{2}}{m}}  =2 e^{-\varepsilon^{2} W_e ^{2} \cdot m} .
		\end{equation*}
			By the union bound we have that 
		\begin{equation}\label{hoefding_degree}
			\Pr[(\exists z \in A_u \cup A_v ) Q(z)] \leq 2 m \cdot 2 e^{-\varepsilon^{2} W_e ^{2} \cdot m}.
		\end{equation}
		Then, by the union bound over \eqref{hoefding_regularity} and \eqref{hoefding_degree} we have that 
		\begin{equation}
			\Pr[(\exists z \in A_u \cup A_v ) Q(z) \textrm{ or } \left(\exists S \subseteq A_u \cup A_v \right)	Q(S)] \leq 
		2 m \cdot 2 e^{-\varepsilon^{2} W_e ^{2} \cdot m} 
			+2^{2 m+1} e^{- W_e^{2} \varepsilon^{2} m^{2} }.
		\end{equation}
		Now, for sufficiently large $m$ we have that 
		\begin{equation}\label{final_graph_equation}
			\Pr[(\exists z \in A_u \cup A_v ) Q(z) \textrm{ or } \left(\exists S \subseteq A_u \cup A_v \right)	Q(S)] <1.
		\end{equation}
		Furthermore, since $W_e \geq W_{\min}$, there is a choice of $m$  such that the equation \eqref{final_graph_equation} holds regardless of the exact choice for $W_e \in [W_{\min},W_{\max}]$, and such that $W_e \cdot \varepsilon\cdot m$  is integer. For such $m$, one can sample a graph $\bar{H}$ from $\mathcal{H}$ such that for each $S \subseteq A_u \cup A_v$ we have 
		\begin{equation*}
			\left|W_{\bar{H}}(S,S) - W_{H}(S,S) \right| < \varepsilon W_e m^{2}
		\end{equation*}
		and for each $z \in A_u \cup A_v$ 	we have $|\deg(z) - W_e m| \leq \varepsilon W_e m$. Let us fix one such graph $\bar{H}$, and match vertices with degree smaller than $(1+\varepsilon) \cdot W_e m $ to obtain a bipartite graph $H'$ in which all vertices have degree $(1+\varepsilon)W_e \cdot m$. Observe that we add at most $ 2 \varepsilon W_e m^{2}$ edges in this step.  
		\par Finally, let us show that $H'$ has the property expressed in the statement of this lemma. Obviously, the graph $H'$ is regular with degrees of vertices being $(1+\varepsilon)W_e m$. It remains to show that for an arbitrary set  $S \subseteq A_u \cup A_v$ we have 
		\begin{equation*}
			\left|W_{H}(S,S) - W_{H'}(S,S) \right| \leq 3 \varepsilon m^{2} W_e. 
		\end{equation*}
		We have that 	
		\begin{equation*}
			\begin{split}
\left|W_{H}(S,S) - W_{H'}(S,S) \right| =   			\left|W_{H}(S,S) - W_{\bar{H}}(S,S) +  W_{\bar{H}}(S,S) - W_{H'}(S,S) \right| \\
				\leq \left|W_{H}(S,S) - W_{\bar{H}}(S,S) \right|+ \left| W_{\bar{H}}(S,S) - W_{H'}(S,S) \right|.
			\end{split}
		\end{equation*}
		Now we have that 	
		\begin{equation*}
			\left|W_{H}(S,S) - W_{\bar{H}}(S,S) \right| \leq \varepsilon W_e m^{2}
		\end{equation*}
since we have insured that the property $Q$  does not hold for any subset $S$  in $\bar{H}$. Furthermore, since we added at most $2\varepsilon W_e m^{2}$  edges we have that 
\begin{equation*}
	\left| w_{\bar{H}}(S,S)-  w_{H'}(S,S) \right| \leq 2 \varepsilon W_e m^{2},
\end{equation*}
which concludes the proof.
	\end{proof}
\end{lemma}
Let us now use this setup to prove Theorem \ref{regular_main_theorem}. The idea is to replace each complete bipartite graph between pairs of sets of vertices $A_u,A_v$ from Lemma \ref{lemma_G'}, with the gadget constructed above, to obtain the regular graph which satisfies \emph{essentially} the same soundness/completeness properties. In particular, we prove the following lemma.
\begin{lemma} \label{main_weight_theorem}
	Let $0<\varepsilon < 1/2, \rho\in (-1,0)$ be rational numbers, and let $\gamma > 0$. Then there is a sufficiently large alphabet size $L \in \mathbb{N}$ and a reduction from Affine Unique Games instances $\Lambda=(\mathcal{U},\mathcal{V},\mathcal{E},\Pi,[L])$ to unweighted simple graphs $G'' = (V', E'')$ with the following properties:
  \begin{itemize}
\item \emph{Completeness:} If $\Opt(\Lambda) \geq 1-\gamma$, then 
	$MSVC(G'') \leq \left(\frac{1}{3-\rho}+\varepsilon+3 \gamma\right) |V'| W_{G''}(E'') + 3 \varepsilon W_{G''}(E'')$.
\item \emph{Soundness:} If $\Opt(\Lambda) \leq \gamma$, then $\MSVC(G'') \geq  \frac{1}{1+\varepsilon} \left(\int_{0}^1 \Gamma_{\rho}(r) dr  -2\varepsilon+o(1)\right)\cdot |V'|W_{G''}(E'')  -3 \varepsilon |V'|W_{G''}(E'')$.
  \end{itemize}
	Moreover, the running time of the reduction is polynomial in $|\mathcal{U}|,|\mathcal{V}|,|\mathcal{E}|$, and exponential in $L$, and the graph $G''$ is regular.
\end{lemma}
\begin{proof}
	For fixed $0<\varepsilon < 1/2, \rho\in (-1,0),\gamma > 0$ let us choose $L$ large enough so that Theorem \ref{hardness_min_sum_vc} holds. Since $\rho<0$ the weights will be in an interval $\left[ \left(\frac{1+\rho}{4}\right)^{L},\left(\frac{1-\rho}{4}\right)^{L}\right]$, and there are finitely many of them. Hence, we can find a constant $q$ for which there will be sufficiently large constant $m$ such that the conditions of Lemma \ref{gadget_lemma} hold for all weights $\left \{ \left(\frac{1+\rho}{4}\right)^{i} \left(\frac{1-\rho}{4}\right)^{L-i}\right\}_{i=0}^{L}$. 
		\par
		Let us replace each edge $e \in E$ with a weighted complete bipartite graph as in Lemma \ref{lemma_G'}, and then replace each bipartite graph with edge weights $W_e$  with a gadget from Lemma \ref{gadget_lemma}. Since there are constantly many different weights, and the size of each gadget is parametrized by $q,\varepsilon,W_{min}$ which are constants, we can deterministicaly find these gadgets in constant time. This concludes the description of the graph $G''$. 
		\par	
		Direct check shows that this reduction is polynomial in $|\mathcal{U}|,|\mathcal{V}|,|\mathcal{E}|$, and exponential in $L$. Let us now show that $G''$ is regular. Let us consider any vertex $v \in V$, and calculate its degree. As stated in Theorem \ref{hardness_min_sum_vc}, there is some value $d = D^{2} 2^{-L+1}$ independent of $v$  such that 
		\begin{equation*}
			\sum_{e \in N(v,V)} W_e =d.
		\end{equation*}
By our construction and the statement of Lemma \ref{gadget_lemma}, each edge of weight $W_e$ incident on a vertex $v$ will be replaced by $W_e m (1+\varepsilon)$  edges. Hence, the degree of $v$  will be 
		\begin{equation*}
			\sum_{e \in N(v,V)} W_e(1+\varepsilon)m= dm(1+\varepsilon)  = D^{2} 2^{-L+1} m (1+\varepsilon).
		\end{equation*}
		 Therefore, the graph $G''$  is indeed regular. 
		\par
		Finally, let us show the soundness and completeness properties of $G''$. Due to Lemma \ref{gadget_lemma}, for each $S \subseteq V'$, we have that 
		\begin{equation*}
\left|W_{G''}(S,S) -W_{G'}(S,S) \right| \leq \sum_{e in E} 3 \varepsilon m^{2} W_e \leq 3 \varepsilon W_{G'}(E') \leq 2 \varepsilon W_{G''}(E'').
		\end{equation*}
		Take any ordering $\sigma$ on the vertex set $V'$ of $G''$. Then, we have that 
			\begin{equation*}
				\SVC_{G''}(\sigma) = \sum_{t=1}^{|V'|} W_{G''}(\sigma([t])^c,\sigma([t])^c).
			\end{equation*}
			This is also an ordering on $G'$, and we can write
\begin{equation} \label{eq:comp_order}
			\begin{split}
				\left|\SVC_{G''}(\sigma)  - \SVC_{G'}(\sigma)\right| = \left|\sum_{t=1}^{|V'|} W_{G''}(\sigma([t])^c,\sigma([t])^c) - W_{G'}(\sigma([t])^c,\sigma([t])^c) \right| \leq 3 \varepsilon |V'| W_{G''}(E'').
			\end{split}
		\end{equation}
Hence, this shows $\MSVC(G'') \leq \MSVC(G') + 3 \varepsilon |V'| W_{G''}(E'')$. Now, we can use completeness from Lemma \ref{lemma_G'} 
	for $\tau=(\frac{1}{3-\rho}+\varepsilon+3 \gamma)$ to show that 
	\begin{equation*}
\MSVC(G') \leq  \left(\frac{1}{3-\rho}+\varepsilon+3 \gamma\right)\cdot W_{G'}(E') |V'| \leq  \left(\frac{1}{3-\rho}+\varepsilon+3 \gamma\right)\cdot W_{G''}(E'') |V'|,
	\end{equation*}
	where we used the fact that $ W_{G''}(E'') \geq  W_{G'}(E')$. Hence, we have that 		
	\begin{equation*}
		\begin{split}
			\MSVC(G'') \leq  \MSVC(G') +  3 \varepsilon |V'| W_{G''}(E'')
			= \left(\frac{1}{3-\rho}+\varepsilon+3 \gamma\right) |V'| W_{G''}(E'') + 3 \varepsilon W_{G''}(E''),
		\end{split}
	\end{equation*}
	which concludes the proof of the completeness property. \par
	For the soundness property we start by proving a lower bound on $\MSVC(G')$. For any ordering $\sigma$ we have 
\begin{equation*}
\SVC_{G'}(\sigma) = \sum_{t=1}^{|V'|} w_{G'}(\sigma([t])^c,\sigma([t])^c) \cdot W_{G'}(E')
\end{equation*}			
	Now, due to Fact \ref{fact:rho_der}, we can apply Lemma \ref{lemma_G'} with $g(r) = \Gamma_{\rho}(r) - \varepsilon$, to conclude that $G'$  is $(r,\Gamma_{\rho}(r)-\varepsilon-1/|V|)$ dense. Since $|V| \geq 2^{L}$, for any sufficiently large $L$  we have that $1/|V| \leq \varepsilon$ and hence the graph $G'$  is
$(r,\Gamma_{\rho}(r)-2\varepsilon)$ dense. Therefore, we have 
	\begin{equation*}
		\begin{split}
		\SVC_{G'}(\sigma)	 \geq  W_{G'}(E') \sum_{t= 1 }^{|V'| }\left( \Gamma_{\rho}(1-t/|V'|) -2\varepsilon \right)
= \left(\int_{0}^{1} \Gamma_{\rho}(1-r) dr -2 \varepsilon +o(1)\right) \cdot |V'| W_{G'}(E') \\
= \left(\int_{0}^{1} \Gamma_{\rho}(r) dr -2\varepsilon + o(1)\right) \cdot |V'| W_{G'}(E').
		\end{split}
\end{equation*}
Due to \eqref{eq:comp_order} we have 
	\begin{equation*}
		\MSVC(G'') \geq  \MSVC(G') - 3 \varepsilon |V'|W_{G''}(E'').
	\end{equation*}
	and therefore
	\begin{equation*}
		\begin{split}
\MSVC(G'') \geq \left(\int_{0}^1 \Gamma_{\rho}(r) dr  -2\varepsilon+o(1)\right)\cdot |V'|W_{G'}(E')  -3 \varepsilon |V'|W_{G''}(E'') \\
			\geq \frac{1}{1+\varepsilon} \left(\int_{0}^1 \Gamma_{\rho}(r) dr  -2\varepsilon+o(1)\right)\cdot |V'|W_{G''}(E'')  -3 \varepsilon |V'|W_{G''}(E'')  \\
		\end{split}
	\end{equation*}
where in the last inequality we used the fact that $(1+\varepsilon)\cdot W_{G'}(E') = W_{G''}(E'').$ 	This concludes our proof.
	\end{proof}
	By Lemma \ref{main_weight_theorem}, by letting $\varepsilon \to 0$, we get inapproximability ratio for MSVC
	\begin{equation*}
		\frac{\int_{0}^1 \Gamma_{\rho}(r) dr } { \frac{1}{3-\rho} } -o(1).
	\end{equation*}
Since we can choose $\rho \in (-1,0)$ arbitrarily, the best approximation ratio that we can obtain is calculated as
\begin{equation*}
	\max_{\rho \in (-1,0) } \frac{\int_{0}^1 \Gamma_{\rho}(r) dr } { \frac{1}{3-\rho} }.
\end{equation*} 
Numerical simulations show that the best inapproximability ratio we can get with these techniques is $1.0157$, and it is obtained for $\rho=-0.52$, as claimed in Theorem \ref{regular_main_theorem}.
\par 
Essentially the same argument shows that the weighted graph from Secton \ref{subsection:non_regular_hardness} can be reduced to the unweighted, albeit not regular, graph, with the similar properties, which yields the proof of Theorem \ref{main_theorem}.
\section{Approximating Min Sum Vertex Cover on Regular Graphs}	\label{section:regular_graphs}
In this section we will revisit an approximation algorithm for Minimum Sum Vertex Cover on regular graphs introduced in \cite{DBLP:journals/algorithmica/FeigeLT04}, in Theorem 11. The authors in that work did not explicitly state the approximation ratio obtained by that algorithm,  since their primary interest was showing that $4/3$-approximation achieved by the greedy algorithm can be beaten by more advanced techniques. 
\par
	We will here give an explicit constant, also taking into account progress in the approximation of the so called Max-$k$-VC problem, which is used in that approach, and for which better algorithms exist since the publication of the aforementioned article. 
	\par

	Before discussing the algorithm, let us define the Max-$k$-VC problem. In this problem a graph $G=(V,E)$  is given as an input, and the goal is to find $S \subseteq V, |S|=k$, such that $w(S,V)$  is as big as possible. Austrin, Benabbas and Georgiou~\cite{DBLP:conf/soda/AustrinBG13} show that Max-$2$-Sat with a bisection constraint, that is, Max-$2$-Sat in which admissible assignments have exactly half of the variables set to $1$, and the other half to $0$, can be approximated within $\alpha_{LLZ} \approx 0.9401$. Let us remark that $ \alpha_{LLZ}$	is the optimal\footnote{Assuming the Unique Games Conjecture} approximation ratio for the Max-$2$-Sat problem \cite{DBLP:conf/ipco/LewinLZ02,DBLP:conf/stoc/Austrin07}. Since this problem subsumes Max-$k$-VC when $k=n/2$, we can approximate Max-$n/2$-VC within $\alpha_{LLZ} \approx 0.9401$. Let us also remark at this point that the Max-$k$-VC problem for $k=n/2$  can not be approximated above $\alpha \geq 0.9431$ \cite{DBLP:conf/approx/AustrinS19}.
	\par
 Let us also recall the following two facts for regular graphs:
	\begin{itemize}
		\item The greedy algorithm on regular graphs covers edges on average by the time $\frac{1}{3}|V|$,
		\item The optimal solution covers an edge on average by the time at least $\frac{1}{4}|V|$.
	\end{itemize} \par
	Let us now discuss the algorithm introduced in \cite{DBLP:journals/algorithmica/FeigeLT04}. We will closely follow the argument outlined there. Let $\varepsilon>0$ be some constant that we will fix later. In case the optimal solution covers an edge by the time $(\frac{1}{4}+\varepsilon)|V|$  later, the greedy algorithm approximates the optimal value within a factor of 
	\begin{equation*}
		\frac{1/3}{\frac{1}{4}+\varepsilon}	 = \frac{4}{3+ 3\cdot 4 \cdot \varepsilon}.
	\end{equation*}
	Otherwise, the optimal solution covers an edge on average at the time $(\frac{1}{4}+\delta) |V|$, for some $\delta \in (0,\varepsilon)$. In this case, we have the following lemma.
	\begin{lemma}\label{reg_lemma}
		Let $G=(V,E)$	 be a regular graph, let $n:=|V|$, and let the optimal solution of Minimum Sum Vertex Cover be $(\frac{1}{4}+\delta)|V|$. Then the optimal solution covers at least $(1-\sqrt{\delta})$ fraction of edges in the first $n/2$ steps.
		\begin{proof}
		Let us denote the degree of the graph  with $D \in \mathbb{N}$, and with $m$ the number of edges $m=n D/2$. We argue by contradiction, and assume that the optimal solution covers less than $(1-\sqrt{\delta})$ fraction of edges in the first $n/2$ steps.  Let us use $u_i,i=1,\hdots,n,$ to denote the number of uncovered edges at the time step $i$, and let $s:=u_{n/2}$. Then by assumption $s > \sqrt{\delta} m $. Furthermore, the value of the minimum sum vertex cover is $\frac{1}{m}\sum_{i=1}^{n}  u_i$. Let us show that $\frac{1}{m}\sum_{i=1}^{n} u_i > (\frac{1}{4} + \delta)n$ yielding a contradiction to the assumption that the optimal solution of Minimum Sum Vertex Cover is $(\frac{1}{4}+\delta)n$.
			\par 
		Let us use $c_i=u_{i}-u_{i-1}$ to denote the number of additionally covered edges at step $i$. Since we are considering the optimal solution to MSVC, the sequence $c_i$ is non-increasing (otherwise changing the order would yield a smaller solution). Furthermore, let us use $c$ to denote $c_{n/2}$.\par
		Now, assuming that $c_{n/2}=c$, $u_{n/2}=s$, let us calculate the smallest possible value of MSVC. We know that after $i$ steps, we can cover at most $i\cdot D$ edges (this happens if all the edges chosen are disjoint). Furthermore, since we assumed that after $n/2$  steps we leave $s$  edges uncovered, and since $c=c_{n/2}$  and $c$  is non-increasing, we have that at the step $i$  we leave at least $s+(n/2-i)\cdot c$ edges uncovered. This shows that 
		\begin{equation*}
			u_i \geq \max\left(\frac{nD}{2}-iD, s+(n/2-i)\cdot c,0\right), i \in [n].
		\end{equation*}
		In particular, the right hand side is a maximum of three linear functions, and therefore, the following scenario for covering the edges will lower bound $u_i$. In the first $t$ fraction of steps, edges get covered at the optimal rate (at each step we cover $D$ new edges), where $t$ is a parameter calculated later. Then, after $t$ fraction of steps, we cover $c$ edges at each step, until we cover all the edges\footnote{In the last step we might cover less than $c$ edges, but we will ignore this case for the sake of simplicity.}.   \par
			\par
			Since we spend $t$ fraction of time covering $D$ edges in each step, the cost of edges covered in this time is
			\begin{equation*}
				\frac{1}{m}\sum_{i=1}^{t\cdot n} \frac{Dn}{2} - i D \geq   n t(1-t).
			\end{equation*}
			The remaining time $x\cdot n$, for some $x \in (0,1)$, is spent on covering $c$ edges at each step. Since after $x$ fraction of steps we covered all the edges, we have 
			\begin{equation*}
				m = t \cdot n \cdot D + x \cdot n \cdot c,
			\end{equation*}
			and since $m = nD/2$	 we have that 
			\begin{equation}\label{x_calc}
				x =\frac{D}{2}	\cdot \frac{1-2t}{c}.
			\end{equation}
			Hence, the average cost of edges incurred in the remaining time is 
			\begin{equation*}
				\frac{1}{m} x \cdot n \cdot (m-t\cdot n \cdot D) \frac{1}{2}= x n \left(\frac{1}{2}-t\right),
			\end{equation*}
			which with 	\eqref{x_calc} yields the cost
			\begin{equation*}
				\frac{D}{2}	\cdot \frac{1-2t}{c} n \left(\frac{1}{2}-t\right)= n \frac{D}{4c}(1-2t)^{2}	.
			\end{equation*}
			Hence, the total cost is 
			\begin{equation}\label{total_cost}
				nt(1-t)+ n \frac{D}{4c}(1-2t)^{2}  
			\end{equation}
			Let us now calculate the value 	of $t$ in terms of $s$ and $c$. We use the fact that after $t$ fraction of steps we covered $t\cdot n \cdot D$ edges, and after  $n/2$ steps we covered  $m-s$ edges. Since we are covering $c$ edges at each step  $i \in [tn,n/2]$ we have that 
			\begin{equation*}
				m-s = t n D + (\frac{n}{2}-tn)\cdot c,
			\end{equation*}
			and from here we get 
			\begin{equation*}
				t = 	\frac{m-s-\frac{n\cdot c}{2}}{n \cdot D - n\cdot c} = \frac{1}{2} - \frac{s}{n (D-c)}.
			\end{equation*}
			Replacing this in \eqref{total_cost} we get that the total cost is at least
				\begin{equation*}
					n\left(	\frac{1}{4}	-\frac{s^{2}}{n^{2}(D-c)^{2}} + \frac{D}{4c} \left( 2\frac{s}{n(D-c)}\right)^{2}\right),
				\end{equation*}
			which reduces to 
			\begin{equation*}
				n \left( \frac{1}{4} + \frac{s^{2}}{n^{2}} \frac{1}{(D-c)c} \right).
			\end{equation*}
			Now, by our contradiction hypothesis we have $s > \sqrt{\delta}m$, and $\frac{1}{(D-c)c} \geq 4\frac{1}{D^{2}}$ (since $c \in (1,D)$), we have that the total cover time is strictly greater than
			\begin{equation*}
				n \left( \frac{1}{4} +  \frac{\delta m^{2} }{n^{2}} \frac{4}{D^{2}}\right) = n \left(\frac{1}{4} + \delta\right),
			\end{equation*}
			which contradicts the fact that the optimal solution to MSVC on the graph $G$ has value $n (\frac{1}{4} + \delta)$. This concludes our proof.
		\end{proof}
	\end{lemma}
In \cite{DBLP:journals/algorithmica/FeigeLT04} it is claimed that in the setup of Lemma \ref{reg_lemma}, the optimal solution covers $(1-\delta)$ fraction of edges. However, this is not correct, and we illustrate that with a counterexample provided in the appendix. Nevertheless, this does not greatly change the conclusion in \cite{DBLP:journals/algorithmica/FeigeLT04}, as using the correct version of the lemma just replaces one unspecified constant below $4/3$ by another unspecified constant below $4/3$.
	\par
	Let us now fix $k=n/2$, and use the Max-$k$-VC algorithm. This will give us a set $S \subseteq V$  such that $w(S,V) \geq  \alpha_{LLZ} (1-\sqrt{\delta})$. We next consider the following ordering of vertices in $V$ and calculate Minimum Sum Vertex Cover for it. We first pick vertices from $S$   in a random order, and then take the remaining vertices in random order as well. Then, the edges in $N(S,V)$   are covered by the time $|V|/4$  in expectation, while the remaining edges are covered by the time $\frac{|V|}{2} + \frac{1}{3} \frac{|V|}{2}$. Hence, we can find an ordering for which Sum Vertex Cover has value
	\begin{equation*}
		w(S,V) 	\frac{|V|}{4} + w(S^{c},S^{c}) ( \frac{|V|}{2} + \frac{|V|}{6} ) \leq  \alpha_{LLZ} (1-\sqrt{\delta}) \frac{|V|}{4} + \left(1-\alpha_{LLZ} (1-\sqrt{\delta}) \right) \cdot \frac{2|V|}{3}. 
	\end{equation*}
	We can simplify this expression as 
	\begin{equation*}
		\left(\frac{-5 \alpha_{LLZ} + 5 \alpha_{LLZ} \sqrt{\delta}}{12} + \frac{2}{3}\right)|V|.
	\end{equation*}
	Hence, we get an approximation ratio of 
	\begin{equation*}
		\frac{\frac{-5 \alpha_{LLZ} + 5 \alpha_{LLZ} \sqrt{\delta} }{12} + \frac{2}{3}}{\frac{1}{4} +\delta}.
	\end{equation*}
	In conclusion, for fixed $\varepsilon$ we have that the approximation ratio is given as
	\begin{equation*}
		\max\left( \frac{4}{3+ 3\cdot 4 \cdot\varepsilon}, \sup_{\delta\in (0,\varepsilon]}		\frac{\frac{-5 \alpha_{LLZ} + 5 \alpha_{LLZ} \sqrt{\delta} }{12} + \frac{2}{3}}{\frac{1}{4} +\delta} \right).
	\end{equation*}
	Optimizing over different values of $\varepsilon$ gives us that the approximation ratio of this algorithm is approximately $1.225$. Let us remark that even the best possible approximation ratio for Max-$k$-VC for $k=n/2$  can not give us an approximation for MSVC better than $1.220$ using the method above.
\section*{Acknowledgements}
The author thanks Johan H\aa stad for helpful discussions and useful comments on the manuscript.
\bibliography{bibl}{}
\bibliographystyle{alpha}
\appendix

\section{Addressing Argument in Theorem 11 in \cite{DBLP:journals/algorithmica/FeigeLT04}}
Let $\delta>0$. We construct a $2$-regular graph $G=(V,E)$ with $|V|=n, |E|=m,$ such that the minimum sum vertex cover has value 
\begin{equation*}
	\left(\frac{1}{4}+\delta\right) n,
\end{equation*}
while any set $S \subseteq V, |S| = \frac{n}{2}$, satisfies $|w(S,V)| \leq (1-\sqrt{\delta}) $. This shows that the factor $(1-\sqrt{\delta})$ in Lemma \ref{reg_lemma} can not be replaced by a sharper $(1-\delta)$, as it was done in \cite{DBLP:journals/algorithmica/FeigeLT04}. \par
Let $t = (\frac{1}{2}-3\sqrt{\delta}) n $  and $s = 2 \sqrt{\delta} \cdot n$, where $n \in \mathbb{N}$  is chosen such that  $t,s \in \mathbb{N}$ (we also approximate $\sqrt{\delta}$ by a rational number), and such that $t$ is even. Then, we construct the graph $G$   by taking $t/2$ disjoint copies of $K_{2,2}$  and $s$  disjoint copies of $K_3$. Let $V_1$  be the set composed of only ``the left sides'' of $t/2$  disjoint copies of $K_{2,2}$, $V_2$  the set composed of only ``the right sides'' of $t/2$   disjoint copies of $K_{2,2}$, and let $V_3$  be composed of the vertices from $s$  disjoint copies of $K_3$.
\par 
Then, the optimal solution for MSVC will work in the following three stages:
\begin{itemize}
	\item \emph{Stage 1:} Pick vertices from $V_1$  in any order. 
	\item \emph{Stage 2:} Pick one vertex from each $K_3$  in $V_3$.
	\item \emph{Stage 3:} Pick another vertex from each $K_3$  in the set $V_3$. 
\end{itemize}
It is clear that this is the minimum sum vertex cover.
The total cost is then split into the following costs:
\begin{itemize}
	\item Edges covered in the first stage. In this case we pick $t\cdot 2 $  edges, and an edge is picked on average at the time $t/2$, so the cost is 
		\begin{equation*}
			\frac{t}{2} \cdot 2 \cdot t = t^{2}.
		\end{equation*}
	\item Edges covered in the second stage. We pick $s \cdot 2$  edges, and each edge is picked on average at the time $t + s /2$, where the factor $t$ exists because this step happens after the first stage. We have the cost of 
		\begin{equation*}
			2\cdot s (t+s/2) = 2\cdot s \cdot t + s^{2}.
		\end{equation*}
	\item Edges covered in the third stage. We pick $s$  edges, and each edge is picked on average at the time $t + s +  s /2$, where the factor $t+s$ exists because this step happens after the second stage. We have the cost of 
		\begin{equation*}
			s (t+s+s/2) =  s t + s^{2} + s^{2}/2.
		\end{equation*}
\end{itemize}
Hence, the total cost is 
\begin{equation*}
	\begin{aligned}
		t^{2}+ 2\cdot s \cdot t + s^{2} + st+s^{2}+s^{2}/2 = t^{2}+3st + \frac{5s^{2}}{2}.
	\end{aligned}
\end{equation*}
Now, recalling that  $t = (\frac{1}{2}-3\sqrt{\delta}) n $  and $s = 2 \sqrt{\delta} \cdot n$, we have that the cost is 
\begin{equation*}
	\begin{aligned}
	n^{2}\cdot \left(\frac{1}{4}-3\sqrt{\delta} + 9 \delta\right) + 3 \cdot 2 \sqrt{\delta} \cdot n \cdot \left(\frac{1}{2}-3\sqrt{\delta}\right) \cdot n + \frac{5\cdot 4  \delta \cdot n^{2}}{2} \\
	= \frac{1}{4}n^{2}-3\sqrt{\delta}n^{2}+ 9 \delta n ^{2} + 3 \sqrt{\delta} n^{2} - 18 \delta n^{2} + 10 \delta n^{2} \\ 
	= \frac{1}{4}n^{2} + \delta n^{2}. 
	\end{aligned}
\end{equation*}
Now, since our graph is $2$-regular graph on $n=2t+3s$  vertices, we have that $m=n$, and we can write the total cost as 
\begin{equation*}
n	\left(\frac{1}{4} + \delta \right),
\end{equation*}
as claimed. It remains to show that for any set $S \subseteq V$ with $|S|=n/2$  we have 
\begin{equation*}
	|E(S,V)| \leq (1-\sqrt{\delta}) \cdot m.
\end{equation*}
It is obvious that the worst case $S$  is exactly the set of vertices picked in the first $n/2$  steps in the algorithm above. In this case, the number of edges not covered is $s/2$, since after $n/2$  steps we are left with one edge uncovered in exactly half of the $K_3$  triangles. Hence, the number of uncovered edges is 
\begin{equation*}
	\frac{s}{2} = \sqrt{\delta} \cdot n =  \sqrt{\delta} \cdot m,
\end{equation*}
and hence 
\begin{equation*}
	|E(S,V)| \leq (1-\sqrt{\delta}) \cdot m,
\end{equation*}
as required.
\end{document}